\documentclass[11pt,reqno]{amsart}

\usepackage{color}
\usepackage[colorlinks=true, allcolors=blue]{hyperref}
\usepackage{amsmath, amssymb, amsthm}
\usepackage{mathrsfs}
\usepackage{mathtools}
\usepackage[noabbrev,capitalize,nameinlink]{cleveref}
\crefname{equation}{}{}
\usepackage{fullpage}
\usepackage[noadjust]{cite}
\usepackage{graphics}
\usepackage{pifont}
\usepackage{tikz}
\usepackage{bbm}
\usepackage[T1]{fontenc}
\DeclareMathOperator*{\argmax}{arg\,max}

\usetikzlibrary{arrows.meta}

\usepackage{environ}
\usepackage{framed}
\usepackage{url}
\usepackage[linesnumbered,ruled,vlined]{algorithm2e}
\usepackage[noend]{algpseudocode}
\usepackage[labelfont=bf]{caption}
\usepackage{cite}
\usepackage{framed}
\usepackage[framemethod=tikz]{mdframed}
\usepackage{appendix}
\usepackage{graphicx}
\usepackage[textsize=tiny]{todonotes}
\usepackage{tcolorbox}
\allowdisplaybreaks[1]

\crefname{algocf}{Algorithm}{Algorithms}

\crefname{equation}{}{} 
\AtBeginEnvironment{appendices}{\crefalias{section}{appendix}} 

\usepackage{enumerate}
\usepackage[color,final]{showkeys} 

\colorlet{refkey}{orange!20}
\colorlet{labelkey}{blue!30}

\crefname{algocf}{Algorithm}{Algorithms}

\numberwithin{equation}{section}
\newtheorem{theorem}{Theorem}[section]
\newtheorem{proposition}[theorem]{Proposition}
\newtheorem{lemma}[theorem]{Lemma}

\crefname{claim}{Claim}{Claims}

\newtheorem{corollary}[theorem]{Corollary}

\newtheorem*{question*}{Question}

\theoremstyle{definition}
\newtheorem{definition}[theorem]{Definition}

\newtheorem*{definition*}{Definition}

\theoremstyle{remark}
\newtheorem*{remark}{Remark}

\newcommand{\set}[1]{\bigg\{ #1 \bigg\}}

\newcommand{\mb}{\mathbb}

\newcommand{\mc}{\mathcal}
\newcommand{\mf}{\mathfrak}

\newcommand{\ol}{\overline}
\newcommand{\on}{\operatorname}
\newcommand{\wh}{\widehat}

\newcommand{\p}{\mathbb{P}}
\newcommand\cal[1]{\mathcal{#1}}

\allowdisplaybreaks
\title{On the sampling Lov\'asz Local Lemma for atomic constraint satisfaction problems}

\author[A1]{Vishesh Jain}
\author[A2]{Huy Tuan Pham}
\author[A3]{Thuy Duong Vuong}
\address{Stanford University, Stanford, CA 94305, USA}
\email{\{visheshj, huypham, tdvuong\}@stanford.edu}

\begin{document}

\begin{abstract}
We study the problem of sampling an approximately uniformly random satisfying assignment for atomic constraint satisfaction problems i.e.~where each constraint is violated by only one assignment to its variables. Let $p$ denote the maximum probability of violation of any constraint and let $\Delta$ denote the maximum degree of the line graph of the constraints.  

Our main result is a nearly-linear (in the number of variables) time  algorithm for this problem, which is valid in a Lov\'asz local lemma type regime that is considerably less restrictive compared to previous works. In particular, we provide sampling algorithms for the uniform distribution on:
\begin{itemize}
    \item $q$-colorings of $k$-uniform hypergraphs with
    \[\Delta \lesssim q^{(k-4)/3 + o_{q}(1)}.\]
    The exponent $1/3$ improves the previously best-known $1/7$ in the case $q, \Delta = O(1)$ [Jain, Pham, Vuong; arXiv, 2020] and $1/9$ in the general case [Feng, He, Yin; STOC 2021]. 
    \item Satisfying assignments of Boolean $k$-CNF formulas with
    \[\Delta \lesssim 2^{k/5.741}.\]
    The constant $5.741$ in the exponent improves the previously best-known $7$ in the case $k = O(1)$ [Jain, Pham, Vuong; arXiv, 2020] and $13$ in the general case [Feng, He, Yin; STOC 2021].
    \item Satisfying assignments of general atomic constraint satisfaction problems with
    \[p\cdot \Delta^{7.043} \lesssim 1.\]
    The constant $7.043$ improves upon the previously best-known constant of $350$ [Feng, He, Yin; STOC 2021].
\end{itemize}
At the heart of our analysis is a novel information-percolation type argument for showing the rapid mixing of the Glauber dynamics for a carefully constructed projection of the uniform distribution on satisfying assignments. Notably, there is no natural partial order on the space, and we believe that the techniques developed for the analysis may be of independent interest. 
\end{abstract}

\maketitle

\section{Introduction}
Let $X_1,\dots, X_{n}$ denote a collection of independent random variables and let $\mc{C} = \{C_1,\dots, C_{m}\}$ denote a collection of events depending on $X_1,\dots, X_n$ (here, the letter $C$ is chosen to represent a ``constraint''). For $C \in \mc{C}$, let $\on{vbl}(C)\subseteq \{X_1,\dots, X_n\}$ be such that $C$ depends only on $X_i \in \on{vbl}(C)$. The celebrated Lov\'asz Local Lemma (LLL) \cite{erdHos1973problems} (stated here in its variable version, symmetric form) asserts that
\begin{equation}
    \label{eqn:LLL-condition}
    e\cdot p\cdot \Delta \leq 1 \implies \mb{P}[\wedge_{i \in [m]}\ol{C_i}] \ge (1 - e\cdot p)^{|\cal{C}|} > 0,
\end{equation}
where $\ol{C}$ denotes the complement of the event $C$, $e$ is the base of the natural logarithm, 
\begin{equation}
\label{eqn:def-p}
p = \max_{i \in [m]}\mb{P}[C_i],
\end{equation} 
and $\Delta \geq 1$ satisfies
\begin{equation}
\label{eqn:def-D}
\#\{j\in [m]: \on{vbl}(C_j)\cap \on{vbl}(C_i)\neq \emptyset\} \leq \Delta \quad\text{for all }i\in [m].
\end{equation}

The original proof of \cref{eqn:LLL-condition} is non-constructive and does not provide an efficient algorithm to find a point in $\wedge_{i \in [m]}\ol{C_i}$ (such a point is called a satisfying assignment). After much work over a period of two decades (cf.~\cite{beck1991algorithmic, alon1991parallel, molloy1998further, czumaj2000coloring, srinivasan2008improved, moser2008derandomizing, moser2009constructive}), the landmark work of Moser and Tardos \cite{moser2010constructive} provided an efficient (randomized) algorithm to find a satisfying assignment whenever the LLL condition (i.e.~the condition on the left hand side of \cref{eqn:LLL-condition}) is satisfied, provided that one is able to efficiently sample from the distribution of $X_i$, and efficiently able to determine the set of constraints that are violated by a given realization of $X_1,\dots, X_n$.\\

In recent years, much attention (cf.~\cite{hermon2019rapid, moitra2019approximate, guo2019counting, guo2019uniform, feng2020fast, feng2020sampling, jain2020towards}) has been devoted to approximate counting and sampling variants of the algorithmic LLL: under conditions similar to the LLL condition, can we efficiently approximately count the total number of satisfying assignments? Can we efficiently sample from approximately the uniform distribution on satisfying assignments?\\ 

This problem turns out to be computationally harder than the problem of efficiently finding one satisfying assignment. Indeed, consider the Boolean $k$-$\on{CNF-SAT}$ problem, in which we are given $n$ Boolean variables $x_1,\dots, x_n$ and $m$ constraints $C_1,\dots, C_m$ such that each constraint depends on exactly $k$ variables, and each constraint is violated by exactly one assignment (out of the $2^{k}$ possible assignments) to its variables. A direct application of the LLL shows that if each constraint shares variables with at most (approximately) $2^{k}/e$ other constraints, then the formula has a satisfying assignment, in which case, the algorithm of \cite{moser2010constructive} efficiently finds such a satisfying assignment. However, it was shown by Bez\'akov\'a et al.~\cite{bezakova2019approximation} that it is $\boldsymbol{\on{NP}}$-hard to approximately count the number of satisfying assignments for a Boolean $k$-CNF formula in which every variable is allowed to be present in $\geq 5\cdot 2^{k/2}$ constraints, even when the formula is monotone.   

On the algorithmic side, both deterministic and randomized algorithms have been devised for approximate counting under LLL-like conditions. On the deterministic side, Moitra \cite{moitra2019approximate} provided a deterministic algorithm to approximately count the number of satisfying assignments of a Boolean $k$-CNF formula in which each constraint shares variables with at most $\Delta \lesssim 2^{k/60}$ other constraints (the $\lesssim$ hides polynomial factors in $k$), provided that $k = O(1)$. Moitra's method was extended by Guo, Liao, Lu, and Zhang \cite{guo2019counting} to provide an efficient deterministic algorithm for approximately counting proper $q$-colorings of $k$-uniform hypergraphs with maximum degree $d$, provided that $q \gtrsim d^{14/(k-14)}$ and $q,k = O(1)$. Recently \cite{jain2020towards}, the authors of this paper showed that for \emph{any} instance of the variable version, symmetric form of the LLL, if each constraint depends on at most $k$ variables and if each variable takes on at most $q$ values, then there is an efficient deterministic algorithm to approximately count the number of satisfying assignments provided that $q,k = O(1)$ and $p \Delta^{7} \lesssim 1$, where $\lesssim$ hides polynomial factors in $q,k$. Here, $p$ and $\Delta$ are as in \cref{eqn:def-p,eqn:def-D}. In particular, this subsumes and improves upon \cite{moitra2019approximate, guo2019counting}. We note that these approximate counting algorithms also lead to efficient algorithms for sampling from approximately the uniform distribution on satisfying assignments. \\

On the randomized side, algorithms have been devised for instances of the variable version, symmetric LLL with \emph{atomic} constraints. Here, an atomic constraint refers to a constraint which is violated by exactly one assignment to its variables. For the special case of \emph{monotone} Boolean $k$-CNF formulas, Hermon, Sly, and Zhang \cite{hermon2019rapid} showed that the Glauber dynamics mix rapidly provided that each variable is present in at most $c2^{k/2}$ constraints, for some absolute constant $c$; note that this matches the hardness regime from \cite{bezakova2019approximation} up to a constant factor. For \emph{extremal} Boolean $k$-CNF formulas (see \cite{guo2019uniform} for the definition) and for $d$ in the entire LLL regime, the method of partial rejection sampling due to Guo, Jerrum, and Liu \cite{guo2019uniform} allows efficient \emph{perfect} sampling from the uniform distribution on satisfying assignments. For general Boolean $k$-CNF formulas, Feng, Guo, Yin, and Zhang \cite{feng2020fast} analyzed the Glauber dynamics on a certain ``projected space'' inspired by Moitra's method, and obtained a near-linear time algorithm for sampling from approximately the uniform distribution on satisfying assignments provided that $\Delta \lesssim 2^{k/20}$ ($\lesssim$ hides polynomial factors in $k$) -- the motivation for constructing the projected space is that, while the original space of satisfying assignments might not even be connected, by passing to an appropriate projection, not only do we have connectivity, but no bottlenecks for the Glauber dynamics. At the same time, since we want to be able to sample from the original distribution conditioned on the realisation of the projected assignment, the projection should be relatively `mild' so as not to lose too much information.\\ 

Most relevant to this paper is the recent work of Feng, He, and Yin \cite{feng2020sampling}, which introduced the idea of `states compression', thereby considerably expanding the applicability of the method used in \cite{feng2020fast}. We will survey their results in the next subsection when we compare them with our own. Here, we only note that compared to the deterministic algorithms for approximate counting, the randomized algorithms discussed here have two advantages: the running time is much faster (in fact, nearly linear in $n$), and they are efficient even when parameters such as $k,q$ grow with $n$. On the other hand, the disadvantage is that so far, these methods are limited to atomic constraints, whereas the algorithm of \cite{jain2020towards} is applicable to general instances of the symmetric LLL.

\subsection{Our results}
\label{sub:results}
We provide randomized algorithms for approximately counting the number of satisfying assignments and sampling from approximately the uniform distribution on satisfying assignments for LLL instances with atomic constraints. 

We begin with our result for the following class of instances, which capture many interesting problems such as Boolean $k$-$\on{CNF-SAT}$ and $k$-hypergraph $q$-coloring. Later, in \cref{thm:unrestricted}, we discuss a result for general atomic constraints. 

\begin{definition}
A $(k,\Delta,q)$-CSP (constraint satisfaction problem) is an instance of the variable version, symmetric LLL in which each variable $X_i$ is uniformly distributed on an alphabet $\Omega_i$ of size $q$, each constraint depends on exactly $k$ variables, and each constraint shares variables with at most $\Delta$ other constraints.  
\end{definition}

As before, we say that a $(k,\Delta,q)$-CSP is atomic if every constraint is violated by exactly one assignment to its variables. Note that for an atomic $(k,\Delta,q)$-CSP, the LLL asserts that if $\Delta \leq cq^{k}$, for an absolute constant $c$, then there exists a satisfying assignment. 

\begin{theorem}
\label{thm:atomic}
Given an atomic $(k,\Delta,q)$-CSP on the variables $X_1,\dots, X_n$, an error parameter $\epsilon \in (0,1/2)$, and a parameter $\eta \in (0,1)$, suppose that one of the following conditions holds.
\begin{enumerate}[(T1)]
    \item $k \geq 3$, $q \geq q_0(\eta)$, and $\Delta \leq c(\eta)\cdot q^{(k-1)/3 + o_{q}(1)}$.
    \item $k = 2$, $q \geq q_0(\eta)$, and $\Delta \leq c(\eta) \cdot q^{4/9 + o_{q}(1)}$.
    \item $k \geq 2$, $q \geq 4$, $\Delta \leq c(\eta)\cdot q^{0.221k}/(k^2\cdot q\log{q})$. 
    \item $k \geq 2$, $q = 3$, $\Delta \leq c(\eta)\cdot 3^{0.2k}/k^2$.
    \item $k \geq 2$, $q = 2$. $\Delta \leq c(\eta)\cdot 2^{0.1742k}/k^2$.
\end{enumerate}
Here $q_0(\eta)$ and $c(\eta)$ are constants depending only on $\eta$. 
Then, there is a randomized algorithm which runs in time
\[\tilde{O}(n\cdot((n/\epsilon)^{\eta}+\Delta)\cdot \Delta\cdot k),\]
where $\tilde{O}$ hides polylogarithmic factors in $n,\Delta, 1/\epsilon$, $k$, $q$, and outputs a random assignment $X \in \prod_{i \in [n]}\Omega_i$ such that the distribution $\mu_{\on{alg}}$ of $X$ satisfies
\[d_{\on{TV}}(\mu_{\on{alg}}, \on{uniform-satisfying}) \leq \epsilon,\]
where $\on{uniform-satisfying}$ denotes the uniform distribution on satisfying assignments and $\on{d}_{\on{TV}}$ denotes the total variation distance between probability measures.  
\end{theorem}

\begin{remark}
In \cite{feng2020sampling}, analogs of (T1) and (T5) are considered. In these cases, they obtain an algorithm for sampling from approximately the distribution $\on{uniform-satisfying}$, and with a similar running time, under the more restrictive conditions:\\

\noindent (T'1) \cite[Theorem~5,4]{feng2020sampling} $k \geq 13$, $q \geq q_0$, and $\Delta \leq q^{(k-12)/9}$.

\noindent (T'5) \cite[Theorem~5.5]{feng2020sampling} $k \geq 2$, $q = 2$, $\Delta \leq c(\eta)\cdot 2^{k/13}$.
\end{remark}

\begin{remark}
We find case (T1) of \cref{thm:atomic} remarkable since, prior to the work of Moser \cite{moser2008derandomizing}, the best-known version of the \emph{existential} algorithmic LLL due to Srinivasan \cite{srinivasan2008improved} required the condition $ p\Delta^{4} \leq c$ (for an absolute constant $c$, and with notation as in \cref{eqn:def-p}, \cref{eqn:def-D}); in particular, \cite{srinivasan2008improved} does not guarantee efficiently finding even a single satsifying assignment in the regime (T1) (for sufficiently large $q$). The chief innovation of Moser was to use denser witness trees instead of so-called $2$-trees (\cref{def:2-tree}); however, in our work, we are able to bypass the $\Delta^4$ barrier for atomic $(k,\Delta,q)$-CSPs, for sufficiently large $q$, even while using $2$-trees.  
\end{remark}

We pause here to record a couple of particularly interesting corollaries of \cref{thm:atomic}. Let $H = (V,E)$ denote a $k$-uniform hypergraph with vertex set $V$ and edge set $E$. Recall that a proper $q$-coloring of $H$ is an assignment $\chi: V \to [q]$ such that for every edge $e$, there exist $u,v \in e$ with $\chi(u) \neq \chi(v)$. In words, no edge is monochromatic. The problem of properly $q$-coloring $H$ can be recast as an atomic $(k,\Delta\cdot q, q)$-CSP, where $\Delta$ denotes the maximum number of edges that any edge of $H$ intersects. Indeed, we simply add $q$ constraints for each edge, where constraint $i$ for the edge is violated if each vertex in the edge is colored with $i$. Then, by (T1), we have:
\begin{corollary}
\label{cor:coloring}
Let $H = (V,E)$ be a $k$-uniform hypergraph with $k \geq 4$ and let $\Delta$ be defined as above. Then, for any $\epsilon, \eta \in (0,1)$, for $q \geq q_0(\eta)$, and for $\Delta \leq c(\eta)\cdot q^{(k-4)/3 + o_q(1)}$, we can sample from a distribution which is $\epsilon$-close in total variation distance to the uniform distribution on proper $q$-colorings of $H$, in time $\tilde{O}(n\cdot((n/\epsilon)^{\eta}+\Delta)\cdot \Delta\cdot k)$.
\end{corollary}

\begin{remark}
This corollary improves upon \cite[Theorem~1.3]{feng2020sampling} which requires $\Delta \leq q^{(k-12)/9 + o_q(1)}$, and on the previous best known regime (even in the bounded degree case) of $\Delta \lesssim q^{(k-4)/7}$ due to \cite{jain2020towards}.
\end{remark}

The next corollary follows from (T5).

\begin{corollary}
\label{cor:sat}
Consider a Boolean $k$-$\on{CNF-SAT}$ instance on $n$ variables $x_1,\dots, x_n$ such that each constraint shares variables with at most $\Delta$ other constraints.  Then, for any $\epsilon, \eta \in (0,1)$ and for $\Delta \leq c(\eta)\cdot 2^{0.1742k}/k^2$, we can sample from a distribution which is $\epsilon$-close in total variation distance to the uniform distribution on satisfying assignments, in time $\tilde{O}(n\cdot((n/\epsilon)^{\eta}+\Delta)\cdot \Delta\cdot k)$.
\end{corollary}

\begin{remark}
The constant $0.1742$ in the exponent is within a factor of less than $3$ of the hardness regime from \cite{bezakova2019approximation}. This improves upon the constant $0.0769$ from \cite[Theorem~1.4]{feng2020sampling} and on the previous best known constant (even in the bounded degree case) of $0.1428$ due to \cite{jain2020towards}.
\end{remark}

We now present our result for general atomic instances of the LLL.

\begin{theorem}
\label{thm:unrestricted}
Given an atomic instance of the LLL, let $k$ denote an upper bound on the number of variables in any constraint, and let $q$ denote an upper bound on the size of the support of any variable $X_i$. Let $\epsilon, \eta \in (0,1)$. Let $\Delta$ be as in \cref{eqn:def-D}, $p \leq p_0(\eta)$ be as in \cref{eqn:def-p}, and suppose that
\[p\cdot \Delta^{7.043 + o_{p}(1)} \leq 1.\]
Then, there is an algorithm which runs in time
\[\tilde{O}(n\cdot((n/\epsilon)^{\eta}+\Delta)\cdot \Delta\cdot k),\]
where $\tilde{O}$ hides polylogarithmic factors in $n,\Delta, 1/\epsilon, k, q$, and outputs a random assignment $X$ such that the distribution $\mu_{\on{alg}}$ of $X$ satisfies
\[d_{\on{TV}}(\mu_{\on{alg}}, \on{uniform-satisfying}) \leq \epsilon.\]
\end{theorem}

\begin{remark}
The constant $7.043$ (which has not been completely optimized and may be slightly lowered) improves upon the constant $350$ from \cite[Theorem~1.1]{feng2020sampling}. Moreover, given a CSP for which every constraint is violated by at most $N$ assignments to its variables, we can construct an atomic CSP with at most $N$ atomic constraints for every original constraint, and thereby obtain a result similar to \cref{thm:unrestricted}, with $\Delta$ replaced by $\Delta N$. We further note that the constant $7.043$ is almost the same as the constant $7$ in \cite{jain2020towards}; while \cref{thm:unrestricted} only applies to atomic CSPs, its advantage is the much faster running time, as well as an LLL type condition which does not depend on $k$ or $q$. 
\end{remark}

\subsection{Approximate counting} 
\cref{thm:atomic,thm:unrestricted} also imply efficient algorithms for approximately counting the number of satisfying assignments in the same regime. Indeed, by using the simulated annealing reduction in \cite{feng2020fast}, one can easily show that if $T(\epsilon)$ is the time to obtain one sample (from a distribution which is $\epsilon$-close in total variation distance to the uniform distribution), then for any $\delta \in (0,1)$, there is a randomized algorithm for approximately counting the number of satisfying assignments within a multiplicative factor of $(1+\delta)$, which runs in time
\[\tilde{O}\left(\frac{m}{\delta^2}T(\epsilon_{m,\delta})\right),\]
where $m$ denotes the number of constraints i.e.~$m = |\mc{C}|$, and
\[\epsilon_{m,\delta} = \Theta\left(\frac{\delta^2}{m\log(m/\delta)}\right).\]

\subsection{Techniques}
\label{sub:techniques}
In \cite{hermon2019rapid}, the authors showed that for a Boolean $k$-CNF formula which is \emph{monotone}, the Glauber dynamics on the space of satisfying assignments mix rapidly outside (a constant factor of) the hardness regime identified in \cite{bezakova2019approximation}. For not necessarily monotone $k$-CNF Boolean formulas, the space of satisfying assignments may not even be connected. To overcome this barrier, and inspired by Moitra's approach \cite{moitra2019approximate} of `marking' variables, the work \cite{feng2020fast} introduced the following two step procedure for sampling a uniformly random satisfying assignment: first, sample from the induced distribution on the so-called unmarked variables, and then, given such a sample $Y$, sample from the uniform distribution on the satisfying assignments conditioned on the assignment to the unmarked variables being $Y$. The reason the last step is easy is that, given a typical assignment to the unmarked variables, the remainder of the formula factors into logarithmic-sized connected components, so that ordinary rejection sampling succeeds with high probability. The key, therefore, is to sample from the induced distribution on unmarked variables. 

The recent work \cite{feng2020sampling} introduced the idea of `states-compression', which generalizes the marking procedure of Moitra. Now, for each variable $v$ with domain $\Omega_v$, one constructs a suitable map $\pi_{v}: \Omega_v \to Q_{v}$, and the assignment $Y$ now lives in $\prod_{v \in V}Q_{v}$. Once again, the projection is to be chosen so that given a typical realisation of $Y$, the remainder of the formula factors into logarithmic-sized connected components on which ordinary rejection sampling succeeds with high probability, and the main part is showing that one can efficiently sample $Y$ from the corresponding distribution.

In order to sample $Y$, both \cite{feng2020fast, feng2020sampling} show that the Glauber dynamics for the distribution on the space $\prod_{v\in V}Q_{v}$, induced by the uniform distribution on satisfying assignments, mix rapidly. For this, both works employ a one-step path coupling argument based on, and extending, the argument from \cite{moitra2019approximate}. However, showing that the one-step path coupling is contracting requires additional assumptions on the relationship between $p$ and $\Delta$, and indeed, this is the main reason for the degradation of the dependence between $p$ and $\Delta$ in the final results of \cite{feng2020fast, feng2020sampling} (and also in \cite{moitra2019approximate, guo2019counting}). Furthermore, establishing contraction of the one-step path coupling requires considerable case analysis for different ranges of the parameters -- for instance, the mixing for the regimes corresponding to \cref{thm:atomic} and \cref{thm:unrestricted} are analyzed separately in \cite{feng2020sampling}.
\\

The main contribution of our work is to completely dispense with the path coupling analysis of this `projected Glauber dynamics' and instead, to devise a novel information-percolation based argument which avoids the need to consider worst-case neighborhoods of a vertex. Such an argument is also the crux of \cite{hermon2019rapid} (which, in turn, is inspired by the argument in \cite{lubetzky2016information}). However, compared to \cite{hermon2019rapid}, we critically do \emph{not} have monotonicity at our disposal. This makes certain `sandwiching' arguments inaccessible, and consequently, necessitate developing a careful and somewhat elaborate notion of combinatorial structures, which we call minimal discrepancy checks (\cref{def:minimal-disc-check}). In a nutshell, the information-percolation argument is based on the fact that if the maximal one-step coupling of the Glauber dynamics fails to couple at some time, then there must already be another discrepancy between the configurations at that time. By tracking the evolution of these discrepancies back in time, we show, in fact, that the origin of the failure of the one-step maximal coupling can be attributed to the appearance of a minimal discrepancy check. By analyzing these minimal discrepancy checks with considerable care, we then show that they occur with probability which is essentially just low enough (\cref{prop:badprob-check}) so as to overcome the union bound on the number of possible minimal discrepancy checks. We expect this part of our argument to also be useful in other contexts. 

Introducing minimal discrepancy checks allows us to handle the projected Glauber dynamics in all regimes in a unified manner. Another component which facilitates this, and also contributes to our improved quantitative estimates, is the notion of admissible projection schemes (\cref{assumptions-1}) -- in contrast to \cite{feng2020sampling}, the conditions we demand of the projections $\pi_v : \Omega_v \to Q_v$ are seemingly more complicated, but these are exactly the conditions which show up in the analysis of the algorithm, and therefore, avoid unnecessary degradation of the parameters. Once the correct conditions for the admissible projection schemes are identified (and this is the non-trivial part), the argument for the existence itself is a standard (although a necessarily quite careful) probabilistic argument. 

\subsection{Organization}
\label{sub:organization}
The remainder of this paper is organized as follows. In \cref{sec:prelims}, we record some preliminary notions related to the Lov\'asz local lemma and constraint satisfaction problems. In \cref{sec:projection}, we introduce admissible projection schemes and state the result (\cref{prop:projection}) guaranteeing the existence of admissible projection schemes in the LLL regime. The proof of \cref{prop:projection}, which is initiated in \cref{sec:projection} is completed in \cref{sec:finish-projection}. In \cref{sec:algorithm}, we present our main sampling algorithm. The main result in this section is \cref{thm:mainsample}, which implies \cref{thm:atomic} and \cref{thm:unrestricted}. The key ingredient required for the proof of \cref{thm:mainsample} is \cref{prop:TV-projection}. This is proved in \cref{sec:glauber}, which is the key section of the paper.

\section{Preliminaries}
\label{sec:prelims}
\subsection{Lov\'asz Local Lemma} 
The LLL provides a sufficient condition guaranteeing that the probability of avoiding a collection $\mc{C}$ of ``bad events'' in a probability space is positive. When the LLL condition \cref{eqn:LLL-condition} is satisfied, the so-called LLL distribution,
\[\mu_{S}[\cdot] := \mb{P}[\cdot \mid \wedge_{C \in \mc{C}}\ol{C}]\]
is well-defined (here, the subscript $S$ is chosen to represent ``satisfying''). For later use, we record a standard comparison between the LLL distribution $\mu_{S}[\cdot]$ and the original distribution $\mb{P}[\cdot ]$ on the probability space i.e.\ the product distribution on $X_1,\dots, X_n$. For any event $B$ in the probability space, let
\[\Gamma(B) = \{C \in \mc{C}: \on{vbl}(B) \cap \on{vbl}(C) \neq \emptyset\}.\]
\begin{theorem}[cf.~{\cite[Theorem~2.1]{haeupler2011new}}]
\label{thm:hss-local-lemma}
Under \cref{eqn:LLL-condition}, for any event $B$ in the probability space, 
\[\mu_{S}[B] \leq \mb{P}[B]\prod_{C \in \Gamma(B)}(1-e\cdot \mb{P}[C])^{-1}.\]
\end{theorem}

We also record here the following algorithmic version of the Lov\'asz Local Lemma, which follows directly from Moser-Tardos algorithm \cite{moser2010constructive} and is also used in \cite{moitra2019approximate, feng2020fast,feng2020sampling}.
\begin{theorem}[\cite{moser2010constructive}]
\label{thm:alg-LLL}
Under \cref{eqn:LLL-condition}, for any $\delta \in (0,1)$, there exists a randomized algorithm which outputs, with probability at least $1-\delta$, a satisfying assignment in time $O(n\Delta k\log(1/\delta))$, where $k=\max_{C\in \cal{C}}|\on{vbl}(C)|$.
\end{theorem}
\begin{proof}
By \cite{moser2010constructive}, under \cref{eqn:LLL-condition}, there exists a randomized algorithm which finds a satisfying assignment in at most $\frac{|\cal{C}|}{\Delta} \le n$ steps in expectation, where each step has time complexity $O(\Delta k)$. By Markov's inequality, if we run this algorithm for $2n$ steps, then with probability at least $1/2$, the algorithm returns a satisfying assignment. The desired conclusion now follows by running $\log(1/\delta)$ independent copies of this algorithm for $2n$ steps each. 
\end{proof}


\subsection{Constraint satisfaction problems}
Let $V$ denote a collection of variables with finite domains $(\Omega_v)_{v \in V}$ satisfying $|\Omega_v| \geq 2$ for all $v \in V$. A constraint on $V$ is a map
\[C: \prod_{v \in V}\Omega_{v} \to \{\on{True}, \on{False}\}.\]
We say that $C$ depends on a variable $v \in V$ if there exist $\sigma_1, \sigma_2 \in \prod_{v \in V}\Omega_{v}$ differing only in $v$ such that $C(\sigma_1) \neq C(\sigma_2)$. For every constraint $C$, we fix $\on{vbl}(C) \subseteq V$ containing all variables that $C$ depends on. 
A constraint satisfaction problem (CSP) is specified by $\Phi = (V, (\Omega_v)_{v \in V}, \mc{C})$, where $\mc{C}$ is a collection of constraints. Given a constraint satisfaction problem, we say that $\sigma \in \prod_{v}\Omega_{v}$ satisfies $\Phi$ if and only if
\[C(\sigma) = \on{True} \quad \text{for all }C \in \mc{C}.\]
We define the degree $\Delta$ of a CSP to be
\[\Delta = \max_{C \in \mc{C}}|\{C' \in \mc{C}: \on{vbl}(C) \cap \on{vbl}(C')\neq \emptyset\}|.\]
We say that $C$ is an atomic constraint if
\[|C^{-1}(\on{False})| = 1.\]
A CSP $\Phi$ is said to be atomic if every $C \in \mc{C}$ is an atomic constraint. Popular examples of atomic constraint satisfaction problems are:
\begin{itemize}
    \item $k$-CNF-SAT. Here, $\Omega_{v} = \{0,1\}$ for all $v \in V$ and $|\on{vbl}(C)| = k$ for all $C \in \mc{C}$.
    \item $k$-Hypergraph $q$-coloring. Let $H = (V, \mc{E})$ denote a $k$-uniform hypergraph. To each vertex $v \in V$, we assign a color in $[q]$ such that no hyperedge is monochromatic. This corresponds naturally to an atomic CSP $\Phi = (V, (\Omega_v)_{v\in V}, \mc{C})$ with $\Omega_{v} = [q]$ for all $v \in V$ and $\mc{C} = \{C_{e, i}: e \in \mc{E}, i \in [q]\}$ where for $\sigma \in [q]^{V}$,
    \[C_{e,i}(\sigma) = \on{False} \iff  \sigma(w) = i \quad \forall w \in e.\]
\end{itemize}  
To every CSP, we associate an instance of the LLL as follows: the random variables are $X_1,\dots, X_{v}$, where each $X_i$ is uniformly distributed on $\Omega_{i}$. To each constraint $C \in \mc{C}$, we associate the event
$$\left\{\sigma \in \prod_{v \in V}\Omega_{v}: C(\sigma) = \on{False}\right\}.$$
We will abuse notation and denote this event by $C$ and the collection of all such events by $\mc{C}$. 

Finally, for a CSP $\Phi$, we let $\mu_{\Phi}$ denote the LLL distribution of the associated LLL instance i.e. $\mu_{\Phi}$ is the uniform distribution on satisfying assignments of $\Phi$. When the underlying CSP is clear from context, we will omit the subscript and denote $\mu_{\Phi}$ simply by $\mu$.

\section{Projection schemes} \label{sec:projection}
\subsection{Preliminaries}
Given a CSP $\Phi = (V, (\Omega_{v})_{v\in V}, \mc{C})$, a projection scheme is a collections of maps
\[\pi_{v} : \Omega_{v} \to Q_{v},\]
where $Q_{v}$ is a finite alphabet with $|Q_{v}| \geq 1$. We will frequently denote the collection $(\pi_{v})_{v\in V}$ simply by $\pi$. We let $\mb{P}_{\pi}$ denote the product distribution on $\prod_{v \in V}Q_{v}$ induced via $\pi$ by the uniform distribution on $\prod_{v\in V}\Omega_{v}$. We also let $\mu_{\pi}$ denote the distribution on $\prod_{v \in V}Q_v$ induced via $\pi$ by $\mu = \mu_{\Phi}$.



Let $\Phi$ be an atomic CSP. Recall that this means that for each $C \in \mc{C}$, there exists some $\boldsymbol{C} \in \prod_{v \in \on{vbl}(C)}\Omega_{v}$ such that $X \in \prod_{v \in V}\Omega_{v}$ does not satisfy $C$ if and only if \[X(v) = \boldsymbol{C}(v)\quad \forall v \in \on{vbl}(C).\]
Given an atomic CSP $\Phi$ and a projection scheme $\pi$, for every $C \in \mc{C}$, we define $\boldsymbol{C}_{\pi} \in \prod_{v \in \on{vbl}(C)}Q_{v}$ by
\[\boldsymbol{C}_{\pi}(v) = \pi_{v}(\boldsymbol{C}(v))\quad \forall v \in \on{vbl}(C).\]
This naturally leads to a CSP $\Phi_{\pi} = (V, (Q_{v})_{v \in V}, \mc{C}_{\pi}),$
where for each $C \in \mc{C}$, there is a constraint $C_{\pi} \in \mc{C}_{\pi}$ such that for $Y \in \prod_{v \in V}Q_{v}$, 
\[C_{\pi}(Y)= \on{False} \iff Y(v) = \boldsymbol{C}_{\pi}(v)\quad \forall v \in \on{vbl}(C).\] 
Motivated by this, for a constraint $C_{\pi} \in \mc{C}_{\pi}$, $v \in \on{vbl}(C_{\pi}):= \on{vbl}(C)$ and $Y \in \prod_{v \in V}Q_{v}$, we say that $Y(v)$ does not satisfy $C_{\pi}$ if and only if $Y(v) = \boldsymbol{C}_{\pi}(v)$.


For a constraint $C\in \cal{C}$, let 
\[
b(C) := \max_{Y\in \prod_{v\in V}Q_v} \prod_{u\in \on{vbl}(C)} \p[X(u) = \boldsymbol{C}(u) | Y],
\] where $Y = \pi(X)$. Equivalently,
\[b(C) = \prod_{u \in \on{vbl}(C)}|\pi_{u}^{-1}(\boldsymbol{C}_{\pi}(u))|^{-1}.\]
Let 
\[
b := \max_{C\in \cal{C}} b(C).
\]
Also, let
\[
q := \max_{v\in V, Y\in \prod_{u\in V}Q_u} d_{\on{TV}}(\mb{P}_{\pi}[\on{value}(v) = \cdot], \mu_\pi[\on{value}(v) =  \cdot \mid Y^{-v}]),
\]
where $Y^{-v}$ denotes the $|V|-1$ dimensional vector obtained by removing $Y(v)$ from $Y$. 

The following useful bound on the conditional marginals of $\mu_{\pi}$ follows from \cref{thm:hss-local-lemma}.
\begin{lemma}
\label{lem:marginal}
Let $\Phi$ be an atomic CSP and let $\pi$ be a projection scheme. 
Suppose that $e\cdot b \cdot \Delta \le 1$. Then for any $v \in V$ and any partial assignment $Z \in \prod_{u \in V \setminus \{v\}}Q_{u}$, $$\mu_\pi[\on{value}(v) =  \cdot \mid Z] \le (1-3b)^{-\Delta}\p_{\pi}[\on{value}(v) = \cdot].$$
\end{lemma}
\begin{proof}
Consider the product distribution $\p_Z$ on $\prod_{v \in V}\Omega_{v}$ where each coordinate $u \in V\setminus \{v\}$ is distributed according to $\p[X(u) = \cdot | \pi_u(X(u))=Z(u)]$ and the $v^{th}$ coordinate is uniformly distributed on $\Omega_{v}$. The $\mb{P}_{Z}$ probability that a constraint $C\in \cal{C}$ is not satisfied is at most $b(C)$ by definition. Let $\mu_{Z,S}$ denote the distribution on satisfying assignments of $\Phi$ induced by $\mb{P}_{Z}$. Then, since $e\cdot b\cdot \Delta \leq 1$, we have by \cref{thm:hss-local-lemma} that 
\begin{align*}
\mu_\pi[\on{value}(v) =  \cdot \mid Z] &= \p[\pi_v(X(v)) = \cdot \mid \pi_u(X(u))=Z(u), X \textrm{ satisfies all $C\in \cal{C}$}] \\
&= \mu_{Z, S}[\pi_{v}(X(v)) = \cdot]\\
&\leq \mb{P}_{Z}[\pi_{v}(X(v)) = \cdot]\prod_{C\in \mc{C}: v \in \on{vbl}(C)}(1-e\cdot \mb{P}_{Z}[C])^{-1} \\
&\le \p_{\pi}[\on{value}(v) = \cdot](1-3b)^{-\Delta} .
\qedhere
\end{align*}
\end{proof}


Let $\Phi$ be an atomic CSP and let $\pi$ be a projection scheme. For each constraint $C\in \cal{C}$, let $\ol{\on{vbl}}(C)$ denote the set of variables $v$ in $C$ for which $|Q_v|>1$. Also, for $C \in \mc{C}$, let
\[\zeta(C) := \max_{v \in \ol{\on{vbl}}(C)} \left(1, \min\left(\frac{(1-3b)^\Delta q}{\p_{\pi}[\on{value}(v)=\boldsymbol{C}_\pi(v)]},2\Delta\right)\right)\]

\subsection{Admissible projection schemes}
The next definition isolates the class of projection schemes we will be interested in.
\begin{definition}
\label{assumptions-1}
Let $\Phi$ be an atomic CSP and let $\pi$ be a projection scheme. Let $\eta\in (0,1/2)$. We say that $\pi$ is \emph{admissible} if 
\begin{enumerate}[({A}1)]
    \item $b \le \eta/(300\Delta)$. 
    \item There exists $\kappa \geq 4\log(3000\Delta)$ such that for any $C\in \cal{C}$,
    \begin{align*}
     |\ol{\on{vbl}}(C)|^{2}\cdot \kappa^2 \cdot \zeta(C) \cdot\prod_{v\in \ol{\on{vbl}}(C)} \left((1-3b)^{-\Delta} \p_{\pi}[\on{value}(v)=\boldsymbol{C}_\pi(v)] + e^{-\kappa/3}\right) \le (60000\Delta)^{-2}.
    \end{align*}
    Furthermore, $\kappa \leq K(\log\Delta + \log{q} + \log{k})$ for a universal constant $K$, for $q = \max_{v \in V}|Q_v|$, and for $k = \max_{C \in \mc{C}}|\on{vbl}(C)|$. 
    \item For any $v\in V$ and $C,C'\in \cal{C}$ with $v\in \on{vbl}(C)\cap \on{vbl}(C')$, 
    \[
        \frac{1}{2}\p_\pi[\on{value}(v)=\boldsymbol{C}_\pi(v)] \le \p_\pi[\on{value}(v)=\boldsymbol{C}'_\pi(v)]\le 2\p_\pi[\on{value}(v)=\boldsymbol{C}_\pi(v)].   
    \]
    \item  For $v\in V$, $\pi(v)$ can be computed in time $K\log|\Omega_v|$, and for any $q \in Q_{v}$ a uniform value in $\pi_v^{-1}(q)$ can be sampled in time $K\log|\Omega_v|$, where $K$ is a universal constant.
\end{enumerate}
\end{definition}
\begin{remark}
Note that the condition $b\le \eta/(300\Delta)$ for $\eta < 1/2$ in (A1) guarantees that 
\[
(1-3b)^{-\Delta}\le 1 + 6b\Delta \le 1 + \eta/50.
\]
\end{remark}

The following is the main result of this section.  
\begin{proposition}\label{prop:projection}
Let $\Phi = (V, (\Omega_{v})_{v \in V}, \mc{C})$ be an atomic CSP. 
Suppose that at least one of the following holds:
\begin{enumerate}
    \item $|\Omega_v| = A  \geq A_0$ for all $v \in V$, $|\on{vbl}(C)| = k \geq 2$ for all $C \in \mc{C}$, and $\Delta \le A^{g(k) - o_A(1)}$, where $A_0$ is a constant depending only on $\eta$, and 
$g(k) = \max\left\{\frac{k-1}{3}, \frac{2k}{9}\right\}$.
    \item $|\Omega_v| = A = 2$ for all $v\in V$, $|\on{vbl}(C)| = k \geq 2$ for all $C \in \mc{C}$, and $\Delta \le cA^{0.1742k}/k^2$ where $c$ is a constant depending only on $\eta$.
    \item $|\Omega_v| = A = 3$ for all $v\in V$, $|\on{vbl}(C)| = k \geq 2$ for all $C \in \mc{C}$, and $\Delta \le cA^{0.2 k}/k^2$ where $c$ is a constant depending only on $\eta$. 
    \item $|\Omega_{v}| = A \geq 4$ for all $v \in V$, $|\on{vbl}(C)| = k \geq 2$ for all $C \in \mc{C}$, and $\Delta \le c A^{0.221(k-1)}/(k^2 \log A)$, where $c$ is a constant depending only on $\eta$. 
    \item $\Delta \le p^{-0.142+o_p(1)}$, where $p \leq c$ for a constant $c$ depending only on $\eta$.
\end{enumerate}
 Then, there exists an admissible projection scheme $\pi = (\pi_v)_{v\in V}$ with $\pi_v : \Omega_v \to Q_v$. 
Moreover, for any $\delta \in (0,1)$, this projection scheme can be constructed, with probability at least $1-\delta$, in time $O(n\Delta k \log(1/\delta))$, 
 where $k = \max_{C\in \cal{C}}|\on{vbl}(C)|$.
\end{proposition}

\begin{proof} 
We give here the complete proof of {\bf Case 1}. The proofs of the remaining cases are deferred to \cref{sec:finish-projection}. 

Let $R = \lfloor A^{2/3} \rfloor$. For each $v\in V$, we let $Q_v = [R]$ and define the projection $\pi_v$ arbitrarily so that the preimage of each element in $Q_v$ has size either $\lfloor A/R\rfloor$ or $\lceil A/R\rceil$. Clearly this projection scheme can be constructed in time $O(1)$, $\pi(v)$ can be computed in time $O(\log A)$, and for any $q \in Q_{v}$, a uniformly random element of $\pi_v^{-1}(q)$ can be returned in time $O(\log A)$. This confirms (A4).

Assuming that $A_0$ is a sufficiently large constant depending on $\eta$ and that $\Delta \le A^{g(k)-o_A(1)} < A^{k/3 - o_A(1)}$, we have
\[
b \le \left(\frac{1}{\lfloor A/R\rfloor}\right)^{k} \le (2A^{-1/3})^k \le (300\Delta/\eta)^{-1},
\]
which confirms (A1).
Moreover, for $v\in \on{vbl}(C)$ and for any $C \in \mc{C}$.
\begin{equation}
\frac{1}{1.5A^{2/3}}\le \frac{\lfloor A/R\rfloor}{A} \leq \p_\pi[\on{value}(v) = \boldsymbol{C}_\pi(v)] \le \frac{\lceil A/R\rceil}{A} \le \frac{1.5}{A^{2/3}}, \label{eq:bound-pv}
\end{equation}assuming again that $A_0$ is a sufficiently large constant. This confirms (A3). 

It remains to verify (A2). Let $\kappa = 12\log(3000(\Delta+A))$. Note that $\kappa \ge 4\log(3000\Delta)$. 
For all $C \in \mc{C}$, we have using $(1-3b)^{-\Delta} \le 1 + 6b\Delta \le 1.01$ that
\[\zeta(C) = \max_{v\in \ol{\on{vbl}}(C)} \left(1,\min\left( \frac{(1-3b)^{\Delta}}{\p_\pi[\on{value}(v) = \boldsymbol{C}_\pi(v)]},2\Delta\right)\right)\le \min\left( 1.5A^{2/3}, 2\Delta\right).\]
Using $(1-3b)^{-\Delta} \le 1.01$ and \cref{eq:bound-pv}, by our choice of $\kappa$, we therefore have 
\begin{align*}
    &\zeta(C) \cdot \prod_{v\in \ol{\on{vbl}}(C)} \left((1-3b)^{-\Delta} \p_\pi[\on{value}(v)=\boldsymbol{C}_\pi(v)] + e^{-\kappa/3}\right)\\
    &\le 2\min(A^{2/3},\Delta) \cdot 2^{k} A^{-2k/3}\\
    &= 2^{k+1} \min(A^{2/3},\Delta) \cdot A^{-2k/3}.
\end{align*}
Therefore, 
\begin{align*}
    &|\ol{\on{vbl}}(C)|^2 \cdot \kappa^2 \cdot \zeta(C) \cdot \prod_{v\in \ol{\on{vbl}}(C)} \left((1-3b)^{-\Delta} \p_\pi[\on{value}(v)=\boldsymbol{C}_\pi(v)] + e^{-\kappa/3}\right)\\
    & \le 2^{k+1}k^2 \cdot (12\log(3000(\Delta+A)))^{2}
    \cdot \min(A^{2/3},\Delta) \cdot A^{-2k/3}.
    \end{align*}
Now, we have the following cases:
\begin{itemize}
    \item $A^{2/3} \leq \Delta \leq A^{k}$. In this case, the right hand side can be bounded by
    \begin{align*}
        2^{k+1}k^2 \cdot (12\log(3000(\Delta+A)))^{2}\cdot A^{-2(k-1)/3} \leq (60000\Delta)^{-2},
    \end{align*}
    provided that $\Delta \leq A^{(k-1 - o_{A}(1))/3}.$
    \item $\Delta  < A^{2/3}$. In this case, the right hand side can be bounded by
    \begin{align*}
        2^{k+1}k^2 \cdot (12\log(3000(\Delta+A)))^{2}\cdot\Delta \cdot A^{-2k/3} \leq (60000\Delta)^{-2}
    \end{align*}    
    provided that $\Delta \leq A^{(2/9 - o_A(1))k}$. \qedhere
\end{itemize}

\end{proof}

\section{The main sampling algorithm}
\label{sec:algorithm}
In this section, we present and analyze our main sampling algorithm $\boldsymbol{\on{Main}}(\Phi, \pi, \epsilon)$.\\ 

Let $\Phi = (V, (\Omega_v)_{v \in V}, \mc{C})$ be an atomic constraint satisfaction problem and let $\pi = (\pi_v)_{v\in V}$ with $\pi_v : \Omega_v \to Q_v$ be an admissible projection scheme. In addition to the notation introduced in \cref{sec:prelims,sec:projection}, we will also make use of the following notation. For a subset $V' \subseteq V$ and a partial assignment $Y \in \prod_{v \in V'}Q_v$, we let $\mc{C}(Y)$ denote the set of constraints which are not satisfied by $Y$. Recall that this is the set of constraints $C_{\pi}$ such that $Y(v) = \boldsymbol{C}_{\pi}(v)$ for all $v \in \on{vbl}(C) \cap V'$. We let $G(Y)$ denote the graph whose vertex set is $\mc{C}(Y)$ and such that $C \neq C' \in \mc{C}(Y)$ are connected if and only if $\on{vbl}(C) \cap \on{vbl}(C') \neq \emptyset$. Also, let $H(Y)$ denote the graph whose vertex set is $V$ and such that $u \neq v \in V$ are connected if and only if there exists some $C \in \mc{C}(Y)$ for which $\{u, v\} \in C$. Finally, for each connected component $H'$ of $H(Y)$, let $\cal{C}(H') = \{C\in \cal{C}(Y):\on{vbl}(C)\subseteq H'\}$. 
\newline 
\newline
The following, which is similar to the algorithm considered in \cite{feng2020fast,feng2020sampling}, is our main sampling algorithm. Throughout, we assume that $\Delta, n \geq c_{\ref{prop:TV-projection}}$, where $c_{\ref{prop:TV-projection}}$ is an absolute constant determined by \cref{prop:TV-projection}. 

\medskip 

$\boldsymbol{\on{Main}(\Phi, \pi, \epsilon)}$: The algorithm takes as input an atomic CSP $\Phi$, an admissible projection scheme $\pi$, and an error parameter $\epsilon \in (0,1/2)$, and outputs either $\on{ERROR}$ or a satisfying assignment $X \in \prod_{v \in V}\Omega_{v}$.

\medskip 

\begin{enumerate}[(M1)]
    \item Initialize $Y_0 \in \prod_{v \in V}Q_{v}$ by sampling $Y_0(v)$ independently and uniformly at random from $Q_v$ for each $v\in V$. 
    \item Let $T = C_{\ref{prop:TV-projection}}\kappa n\log (n\Delta/\epsilon)$, where $C_{\ref{prop:TV-projection}}$ is an absolute constant determined by \cref{prop:TV-projection}. For each $t\in \{1,\dots,T\}$, given $Y_{t-1}$, generate $Y_{t}$ by choosing $v_t\in V$ uniformly at random, setting $Y_{t}(u) = Y_{t-1}(u)$ for all $u \neq v_t$, and setting $Y_t(v_t) = \on{Sample}(Y_{t-1},v)$.
    \item Return $\on{InvSample}(Y_{T})$. 
\end{enumerate}

\medskip

The main result of this section, which, together with \cref{prop:projection}, immediately implies \cref{thm:atomic,thm:unrestricted} is the following. 
\begin{theorem}
\label{thm:mainsample}
Let $\Phi$ be an atomic CSP and let $\pi$ be an admissible projection scheme as in \cref{assumptions-1}. Let
\[q = \max_{v \in V}|\Omega_v|, \quad k = \max_{C\in \mc{C}}|\on{vbl}(C)|.\] 
Then, for any $\epsilon \in (0,1/2)$, $\on{Main}(\Phi, \pi, \epsilon)$ runs in time 
\[\tilde{O}(n\cdot \left((n/\epsilon)^{\eta} + \Delta\right)\cdot \Delta \cdot k),\]
where $\tilde{O}$ hides polylogarithmic factors in $n, \Delta, 1/\epsilon, k, q$, 
and outputs either a random satisfying assignment of $\Phi$ or $\on{ERROR}$. Denoting by $\mu_{\on{alg}}$ the distribution produced by the algorithm and by $\mu_\Phi$ the uniform distribution on satisfying assignments of $\Phi$ (trivially extended to take on the value $\on{ERROR}$ with probability $0$), we have that  
\[d_{\on{TV}}(\mu_{\on{alg}},\mu_{\Phi}) \le \epsilon.\]
\end{theorem}

The algorithm $\on{Main}(\Phi, \pi, \epsilon)$ uses two subroutines, $\on{Sample}(Y_{t-1}, v)$ and $\on{InvSample}(Y_T)$, which we now describe in detail.\\

$\boldsymbol{\on{Sample}(Y, v)}$: This subroutine takes as input $Y \in \prod_{v \in V}Q_v$ and $v\in V$, and returns an element of $Q_v$. 

\medskip 

Consider the partial assignment $Y^{-v}$ and the corresponding graph $H(Y^{-v})$. Let $H_{v}$ denote the (maximal) connected component of $v$ in $H(Y^{-v})$. 
\begin{enumerate}[(S1)]
    \item If $|\cal{C}(H_v)| > 20 \Delta \log ( n\kappa/\epsilon)$, output a uniformly random element of $Q_v$. If we are in this case, we say that $\on{Sample}(Y,v)$ fails due to (S1).
    \item Otherwise, $|\mc{C}(H_v)| \leq 20\Delta \log(n\kappa /\epsilon)$. Let $S = 10(\kappa n/\epsilon)^{\eta}\log(n \kappa/\epsilon)$. 
    
    For each $s = 1,\dots ,S$, do the following:
    
    \begin{itemize}
        \item For each $u\in H_v$, $u \neq v$, sample independently and uniformly at random a value $X(u) \in \pi_u^{-1}(Y(u))$. Sample independently and uniformly at random from $\Omega_v$ a value $X(v)$. Denote the resulting $|H_v|$-dimensional vector by $X(H_v)$. 
        \item If $X(H_v)$ satisfies $\cal{C}(H_v)$, then terminate and output $\pi_v(X(v))$. Otherwise, (i) if $s = S$, then go to the next bullet point, (ii) if $s < S$, then skip the next bullet point and increment $s$ by $1$. 
        \item If we reach this bullet point (i.e.~$X(H_v)$ does not satisfy $\mc{C}(H_v)$ for all $s = 1,\dots, S$), then output a uniformly random element of $Q_v$. In this case, we say that $\on{Sample}(Y,v)$ fails due to (S2). 
    \end{itemize}
\end{enumerate}

\medskip

$\boldsymbol{\on{InvSample}(Y)}$: This subroutine takes as input $Y \in \prod_{v\in V}Q_v$ and returns either $\on{ERROR}$ or an assignment in $\prod_{v\in V}\Omega_v$. 

\medskip 

Consider $Y \in \prod_{v\in V}Q_v$ and the graph $H(Y)$. 
\begin{enumerate}[({I}1)]
    \item If any (maximal) connected component $H'$ of $H(Y)$ has $|\cal{C}(H')| > 20 \Delta \log (n\kappa/\epsilon)$, output $\on{ERROR}$. 
    In this case, we say that $\on{InvSample}$ fails due to (I1). 
    \item Otherwise, for each (maximal) connected component $H'$ of $H(Y)$, we have $|\mc{C}(H')| \leq 20\Delta \log( n\kappa/\epsilon)$. Let $S = 10(\kappa n/\epsilon)^{\eta}\log(n \kappa/\epsilon)$. Let $H_1,\dots, H_\ell$ denote an enumeration of the (maximal) connected components of $H(Y)$. 
    
    \noindent For $j  = 1,\dots, \ell$, do the following:
    
    \quad For each $s = 1,\dots, S$, do the following:
    \begin{itemize}
        \item  For each $u\in H_j$, sample independently and uniformly at random a value $X(u) \in \pi_u^{-1}(Y(u))$. Denote the resulting $|H_j|$-dimensional vector by $X'(H_j)$.  
        \item If $X'(H_j)$ satisfies $\cal{C}(H_j)$, then set $X|_{H_j} = X'(H_j)$. Return to the outermost for loop (in $j$) with $j$ incremented by $1$. 
        \item If $X'(H_j)$ does not satisfy $\mc{C}(H_j)$ and $s < S$, then return to the for loop in $s$ with $s$ incremented by $1$. 
        \item If $X'(H_j)$ does not satisfy $\mc{C}(H_j)$ and $s=S$, then terminate both for loops and return $\on{ERROR}$.
        In this case, we say that $\on{InvSample}(Y)$ has failed for $H_j$. 
    \end{itemize}
    \item Output $X$. 
\end{enumerate}
We say that $\on{InvSample}(Y)$ fails due to (I2) if it fails for any connected component $H_1,\dots, H_\ell$.  

\subsection{The distribution of $Y_t$}
The main step in the analysis of the algorithm is the following proposition, which shows rapid mixing for the Glauber dynamics for the distribution $\mu_{\pi}$. Compared to the works \cite{feng2020fast, feng2020sampling}, we are able to establish rapid mixing of the Glauber dynamics for a much wider class of projection schemes -- this is done by abandoning the path coupling approach of \cite{feng2020fast, feng2020sampling} and instead devising an information-percolation type argument extending the approach in \cite{hermon2019rapid} from the monotone Boolean case to general finite alphabets and without any monotonicity assumption.
\begin{proposition}\label{prop:TV-projection}
There exist absolute constants $c_{\ref{prop:TV-projection}}, C_{\ref{prop:TV-projection}} \geq 1$ for which the following holds. 
Let $(Y_{t})_{t\geq 0}$ denote the Glauber dynamics for $\mu_{\pi}$ starting at an arbitrary initial state $Y_0$. Then, for any $\delta \in (0,1/2)$ and for $T = C_{\ref{prop:TV-projection}} \kappa n\log(n\Delta/\delta)$, the total variation distance between the distribution of $Y_T$ and the distribution $\mu_\pi$ is at most $\delta^{4}$, provided that $n \geq c_{\ref{prop:TV-projection}}$. 
\end{proposition}

The proof of this key proposition is the content of \cref{sec:glauber} and constitutes the bulk of this paper. 

\subsection{Connected components of the projected CSP}
To control failures due to (S1) and (I1), we will show that, with high probability, the projected chain $Y_t$ satisfies the property that 
for every connected component $H'$ of $H(Y_t)$, we have $|\cal{C}(H')| \le 20 \Delta\log (n\kappa/\epsilon)$. As in \cite{moitra2019approximate,feng2020fast, feng2020sampling,jain2020towards}, our analysis uses $2$-trees, which were first used in a similar context by Alon \cite{alon1991parallel}.
\begin{definition}
\label{def:2-tree}
Let $G = (V,E)$ be a graph and let $d_G(\cdot, \cdot)$ denote the graph metric. A set of vertices $\cal{T} \subseteq V$ is called a $2$-tree if for every $u \neq v\in \cal{T}$, $d_G(u,v) \ge 2$, and such that if we add an edge between all pairs of vertices $u,v\in \cal{T}$ with $d_G(u,v) \le 2$, then the resulting graph on $\cal{T}$ is connected. 
\end{definition}

We will need the following result on the number of $2$-trees of a prescribed size which contain a given vertex.
\begin{lemma}[cf.~{\cite[Corollary~5.7]{feng2020fast}}]
\label{lem:2-tree-count}
Let $G = (V,E)$ be a graph with maximum degree $\Delta$. For any $v\in V$, the number of $2$-trees in $G$ which contain $v$ and have size $\ell$ is at most $\frac{(e\Delta^2)^{\ell-1}}{2}$. 
\end{lemma}
Here, $e$ is the base of the natural logarithm.\\

The next (standard) lemma shows that large $2$-trees exist in graphs of bounded maximum degree. We include the proof for the reader's convenience.  
\begin{lemma}\label{lem:large-2-tree}
Let $G = (V,E)$ be a graph with maximum degree $\Delta$. Let $H = (V(H), E')$ be a connected subgraph of $G$ and let $v \in V(H)$. Then, there exists a $2$-tree $\cal{T}$ with $v \in \mc{T} \subseteq V(H)$ such that $|\mc{T}| \geq |V(H)|/(\Delta+1)$. 
\end{lemma}
\begin{proof}
We construct such a $2$-tree greedily. Let $\cal{T}_0 = \{v\}$, $H_0 = V(H)$, and $v_0 = v$. In the $i^{th}$-step, for $i\ge 1$, let $H_i = H_{i-1} \setminus (\{v_{i-1}\}\cup N_G(v_{i-1}))$, and choose (if possible) $v_i$ to be some vertex in $H_i$ such that $d_G(\cal{T}_{i-1},v_i)=2$. We then let $\cal{T}_i = \cal{T}_{i-1}\cup \{v_i\}$. Observe that $H_i = V(H) \setminus (\cal{T}_{i-1} \cup N_G(\cal{T}_{i-1}))$. 

We claim that if we cannot find $v_i$ satisfying $d_G(\mc{T}_{i-1}, v_i) = 2$, then it must be the case that $H_i = \emptyset$. Indeed, assume that $H_i$ is nonempty. Since $H_i = V(H) \setminus (\cal{T}_{i-1} \cup N_G(\cal{T}_{i-1}))$ is not empty, and since $H$ is connected, there must exist some $u \in H_i$ for which $d_G(\cal{T}_{i-1},u) > 1$ and some $u' \in \cal{T}_{i-1} \cup N_G(\cal{T}_{i-1})$ such that there is an edge between $u$ and $u'$. Note that $u' \notin \mc{T}_{i-1}$, since $u' \in \mc{T}_{i-1}$ means that $u \in N_G(\mc{T}_{i-1})$, a contradiction. Hence, $d_G(\mc{T}_{i-1}, u') = 1$ so that $d_G(\mc{T}_{i-1}, u) \leq 2$, which combined with the previous lower bound gives $d_G(\mc{T}_{i-1}, u) = 2$. 

Thus, when our construction terminates, we have a $2$-tree $\mc{T}$ satisfying $\cal{T} \cup N_G(\cal{T})\supseteq V(H)$. Since the maximum degree in $G$ is $\Delta$, it follows in particular that $|\cal{T}| \ge |V(H)|/(\Delta+1)$, as desired. 
\end{proof}

With this preparation, we are ready to bound the probability of failure due to (S1) or (I1).

\begin{proposition} \label{prop:connected-component}
Fix $0\leq t \leq T$. Then, with probability at least $1-(\epsilon/\kappa n)^{10}$, for every connected component $H'$ of $H(Y_t)$, we have $|\cal{C}(H')| < 20\Delta\log ( n \kappa /\epsilon)$.
\end{proposition}

\begin{proof}
By the law of total probability, it suffices to prove the result after conditioning on the choice of variables $v_1,\dots, v_t$ chosen to be updated in the first $t$ steps.  
Suppose for contradiction that there exists a connected component $H'$ of $H(Y_t)$ with $|\mc{C}(H')| \geq 20\Delta\log(n \kappa/\epsilon) =: \alpha$.  Then, by definition, there must exist a connected component of $G(Y_t)$ of size at least $\alpha$ containing a constraint $C_* \in \mc{C}(H')$ and a variable $v_*$ with $v_*\in \on{vbl}(C_*)$ (in particular, $v_* \in H'$). Since the maximum degree of $G(Y_t)$ is $\Delta$, it follows from \cref{lem:large-2-tree} that there exists a $2$-tree $\cal{T}_*$ of constraints in $\cal{C}(Y_t)$ such that $C_* \in \mc{T}_*$ and $|\mc{T}_*| \geq \alpha/(\Delta+1)$. We now proceed to bound the probability of appearance of such a $2$-tree.\\ 

By \cref{lem:2-tree-count}, the number of $2$-trees in $G(Y_t)$ which are rooted at $C_*$ and which have size $\ell$ is at most $(e\Delta^2)^{\ell-1}$. The main observation is the following: fix any such $2$-tree $\mc{T}$. Then, by definition, for every constraint $C \in \mc{T}$ and for every $v \in \on{vbl}(C)$, we must have $Y_{t}(v) = \boldsymbol{C}_{\pi}(v)$. Moreover, for any variable $v$, letting $t_{v}$ denotes the last time before (and including) $t$ that the value of $v$ was updated (note that $t_v$ is determined by our conditioning), we have (by the assumed distribution of $Y_0$ and the description of the subroutine $\on{Sample}(Y,u)$) that the value of $v$ at $t_{v}$ is chosen from one of two distributions:
\begin{itemize}
    \item The uniform distribution on $Q_{v}$, or
    \item The distribution $\mu_{\pi}[\on{value}(v) = \cdot \mid Y_{t_v - 1}^{-v}]$.
\end{itemize}
In either case, the probability that $Y_{t_v}(v) = \boldsymbol{C}_{\pi}(v)$ is at most $(1-3b)^{-\Delta}\mb{P}_{\pi}[\on{value}(v) = \boldsymbol{C}_{\pi}(v)]$ -- indeed, this is true for the uniform distribution even without the $(1-3b)^{-\Delta}$ factor, whereas for the other case, this follows immediately from \cref{lem:marginal}. Since the last time before (and including) $t$ that each variable is updated is determined by our conditioning, and since different constraints in $\mc{T}$ are disjoint, it follows that the probability that none of the constraints in $\mc{T}$ are satisfied is at most
\[\prod_{C \in \mc{T}}\prod_{v\in \on{vbl}(C)}((1-3b)^{-\Delta}\p[Y(v)=\boldsymbol{C}_\pi(v)]) \leq \prod_{C \in \mc{T}}(3000\Delta)^{-2} \leq (3000\Delta)^{-2\ell},\]
where the first inequality uses condition (A2) of \cref{assumptions-1}. Therefore, by the union bound, it follows that the probability that there exists some $C_* \in \mc{C}(Y_t)$ and a $2$-tree $\mc{T}_*$ rooted at $C_*$ such that $|\mc{T}_*| \geq \alpha/(\Delta+1) := \ell$, and such that none of the constraints in $\mc{T}_*$ are satisfied, is at most 
\[n\Delta \cdot (e\Delta^2)^{\ell-1}\cdot (3000\Delta)^{-2\ell},\]
where the first factor accounts for the number of choices for $\mc{C}_*$ and the second factor is from \cref{lem:2-tree-count}.
Since by our choice of $\ell$,
\[n\Delta \cdot (e\Delta^{2})^{\ell-1}\cdot (3000\Delta)^{-2\ell} \leq (\epsilon/\kappa n)^{10},\]
we have the required assertion. \qedhere


\end{proof}

\subsection{Rejection sampling}
Having seen that the probability of failure due to (S1) or (I1) is very low, we now show that the probability of failure due to (S2) or (I2) is also very low. 
\begin{proposition}
\label{prop:rejection}
Let $V' \subseteq V$ and let $ Y \in \prod_{v\in V'}Q_v$ be a partial assignment. Let $H'$ be a connected component of $H(Y)$ and suppose that the size of $H'$ is at most $20 \Delta \log (n\kappa/\epsilon)$. Let $X$ be obtained by sampling each $X(v)$ independently and uniformly from $\pi_v^{-1}(Y(v))$ for each $v\in H'\cap V'$, and from $\Omega_v$ for each $v \in H' \cap (V\setminus V')$. Then, the probability that $X$ satisfies all constraints $C \in \mc{C}(H')$ is at least $(n\kappa/\epsilon)^{-\eta}$.
\end{proposition}

Here, $\eta$ is the parameter appearing in \cref{assumptions-1}.

\begin{proof}
By definition of $X$, the probability that a constraint $C \in \mc{C}(H')$ is not satisfied by $X$ is at most 
$$b(C) = \prod_{v\in \on{vbl}(C)} \frac{1}{|\pi_v^{-1}(\boldsymbol{C}_\pi(v))|}.$$
Let $\mc{G}$ denote the event that $X$ satisfies all constraints $C \in \mc{C}(H')$. Since $b(C) \leq b$ for all $C$ and since $300\Delta b \leq \eta$ by (A1) of \cref{assumptions-1}, it follows from \cref{thm:hss-local-lemma} and (A1) that
\[
\mb{P}[\mc{G}] \geq (1-3b)^{|\cal{C}(H')|} \ge (1-3b)^{20 \Delta \log (n\kappa/\epsilon)} \ge (n \kappa/\epsilon)^{-\eta}. \qedhere
\]

\end{proof}

We immediately obtain the following corollary. 
\begin{corollary}\label{cor:failure-rejection}
Fix $1 \leq t \leq T$. The probability that $\on{Sample}$ fails due to (S2) at time $t$ is at most $(\epsilon/\kappa n)^{10}$. Moreover, the probability that $\on{InvSample}$ fails due to (I2) is at most $(\epsilon/ \kappa n)^9$.
\end{corollary}
\begin{proof}
Recall that $S = 10(n\kappa /\epsilon)^{\eta}\log(n \kappa/\epsilon)$.
By \cref{prop:rejection}, the probability that $\on{Sample}$ fails due to (S2) at time $t$ is at most
\[
(1-(\kappa n/\epsilon)^{-\eta})^S \le \exp\left(-10(\kappa n/\epsilon)^{-\eta} (\kappa n/\epsilon)^{\eta}\log(n \kappa/\epsilon)\right) = (\epsilon/\kappa n)^{10}.
\]
Moreover, again by \cref{prop:rejection}, the probability that $\on{InvSample}$ fails due to any connected component $H'$ with $|\cal{C}(H')| \le 20 \Delta \log(n \kappa/\epsilon)$ is at most
\[
(1-(\kappa n/\epsilon)^{-\eta})^S \le (\epsilon/\kappa n)^{10}.
\]
Therefore, by the union bound over all (at most $n$) maximal connected components of $H(Y)$, the probability that $\on{InvSample}$ fails due to (I2) is at most $(\epsilon/\kappa n)^9$.
\end{proof}

\subsection{Analysis of the main algorithm} The proof of \cref{thm:mainsample} now follows readily. 
\begin{proof}[Proof of \cref{thm:mainsample}]
Let $\mu_{\on{alg}}$ be the distribution on \[\{\text{satisfying assignments of }\Phi\} \cup \{\on{ERROR}\}\] 
given by the output of the algorithm $\on{Main}(\Phi, \pi, \epsilon)$. Also, let $\mu_{\on{alg}'}$ be the distribution on \[\{\text{satisfying assignments of }\Phi\} \subseteq \prod_{v \in V}\Omega_v\]  given by 
\[\mu_{\on{alg}'}[x] := \mu_{S}[X = x \mid \pi(X) = Y_T],\]
where $Y_T$ is generated by running the Glauber dynamics for $\mu_\pi$ for $T$ steps, starting from $Y_0$ (where $Y_0$ is as in (M1)). 

The relation between the distributions $\mu_{\on{alg}}$ and $\mu_{\on{alg'}}$ is as follows: let $\mc{G}_T$ denote the event that none of the calls to $\on{Sample}$ fail (either due to (S1) or (S2)) and that the call to $\on{InvSample}$ also does not fail (either due to (I1) or (I2)). Observe that
\[\mu_{\on{alg}} \mid \mc{G}_T = \mu_{\on{alg}'}.\]
Therefore, by the characterization of the total variation distance in terms of coupling (cf.~\cite[Proposition~4.7]{levin2017markov}), we have
\begin{align*}
    d_{\on{TV}}(\mu_{\on{alg}}, \mu_{\on{alg}'})
    &\leq \Pr[\mc{G}_T]\\
    &\leq T(\epsilon/\kappa n)^{10} + (\epsilon/\kappa n)^{9} +  (T+1)(\epsilon/\kappa n)^{10}\\
    &\leq (\epsilon/\kappa n)^{5},
\end{align*}
where the second line follows from 
\cref{prop:connected-component} and \cref{cor:failure-rejection}, and the third line follows from the value of $T$ and since $\kappa \geq \log(\Delta)$ ((A2) of \cref{assumptions-1}).

Moreover, by \cref{prop:TV-projection}, we know that
\[d_{\on{TV}}(\mu_{\pi}, Y_T) \leq \epsilon^{4},\]
from which it immediately follows (again by \cite[Proposition~4.7]{levin2017markov}) that
\[d_{\on{TV}}(\mu_{\on{alg}'}, \mu_\Phi) \leq \epsilon^{4}.\]
Therefore, by the triangle inequality we have that
\[d_{\on{TV}}(\mu_{\on{alg}}, \mu_\Phi) \leq \epsilon.\]

\medskip 

\medskip 

It remains to analyze the running time of the algorithm. Let 
\[q = \max_{v \in V}|\Omega_v|, \quad k = \max_{C\in \mc{C}}|\on{vbl}(C)|.\] 
Each call to $\on{Sample}$ takes time 
    \[\tilde{O}(\left((n/\epsilon)^{\eta} + \Delta\right)\cdot \Delta \cdot k\cdot\log{q}),\]
    where $\tilde{O}$ hides polylogarithmic factors in $n, \Delta, \epsilon^{-1}$. This is because we require $\tilde{O}(\Delta^{2}\cdot k \cdot \log{q})$ time for checking whether or not $|\mc{C}(H_v)| \leq 20\Delta \log(n\kappa/\epsilon)$, and in case the upper bound holds, then finding this component. In the latter case, for each iteration, we require time $\tilde{O}(k\cdot\Delta \cdot \log{q})$ to sample $X(H_v)$ and time $\tilde{O}(k\cdot \Delta\cdot \log{q})$ to check whether $X(H_v)$ satisfies $\mc{C}(H_v)$. Therefore, (M1) and (M2) take time
    \[\tilde{O}(n\cdot \left((n/\epsilon)^{\eta} + \Delta\right)\cdot \Delta \cdot k\cdot\log{q})\]
    
    Moreover, by a similar analysis as for $\on{Sample}$, the call to $\on{InvSample}$ also takes time 
       \[\tilde{O}(n\cdot \left((n/\epsilon)^{\eta} + \Delta\right)\cdot \Delta \cdot k\cdot\log{q}),\]
       so that the running time of the algorithm is
           \[\tilde{O}(n\cdot \left((n/\epsilon)^{\eta} + \Delta\right)\cdot \Delta \cdot k\cdot\log{q}),\]
           as desired. \qedhere
\end{proof}

\section{Glauber dynamics for the projected distribution: Proof of \cref{prop:TV-projection}}
\label{sec:glauber}
\subsection{Preliminaries}
Throughout this section, we fix an atomic CSP $\Phi = (V, (\Omega_v)_{v\in V}, \mc{C})$ and an admissible projection scheme $\pi = (\pi_v)_{v\in V}$ with $\pi_v : \Omega_v \to Q_v$. 
Recall that $\mu_{\pi}$ is the distribution on $\prod_{v\in V}Q_v$ induced via $\pi$ by the uniform distribution on satisfying assignments, $\mu = \mu_{\Phi}$ on $\prod_{v\in V}\Omega_v$. In this section, which is the main innovation of our work, we study the mixing of the Glauber dynamics for the distribution $\mu_{\pi}$. Recall that the Glauber dynamics is a discrete time Markov chain on the state space $\prod_{v\in V}Q_{v}$ whose transitions are as follows: given the current state $Y$, choose a uniformly random vertex $v\in V$ and move to the state $Y'$ where
\begin{align*}
    Y'(w) &= Y(w) \quad \forall w\neq v\\
    Y'(v) &\sim \mu_{\pi}[\on{value}(v)= \cdot \mid Y^{-v}].
\end{align*}
It is standard that this chain is aperiodic and reversible with respect to $\mu_{\pi}$ and by using the condition $e\cdot b\cdot \Delta < 1$ along with the LLL, it is also easily seen (cf.~\cite[Proposition~8.1]{feng2020sampling}) that this chain is irreducible. Therefore (cf.~\cite[Corollary~1.17]{levin2017markov}), $\mu_{\pi}$ is the unique stationary distribution of this chain. 

We denote the Glauber dynamics for $\mu_{\pi}$ by $(Z_t)_{t\geq 0}$. 
Given $X_0, Y_0$, let $(X_t, Y_t)_{t\geq 0}$ denote a coupling of two copies of $Z_t$ starting from $X_0$ and $Y_0$. For this coupling, let 
\[\tau_{\on{couple}} = \min\{t\geq 0: X_{t} = Y_{t}\}.\]
It is well known (cf.~\cite[Theorem~5.4]{levin2017markov}) that
\begin{equation}
\label{eqn:mixing}
\max_{Z_0}d_{\on{TV}}(Z_{t}, \mu_{\pi}) \leq \max_{X_0, Y_0}\inf_{\on{couplings}} \mb{P}[\tau_{\on{couple}}\geq t],
\end{equation}
where the infimum is taken over all couplings $(X_t)_{t\geq 0}$ and $(Y_t)_{t\geq 0}$ of two copies of the Glauber dynamics with initial states $X_0$ and $Y_0$ respectively. Thus, our goal in this section is to show that for any $X_0, Y_0$, there is a coupling $(X_t, Y_t)_{t\geq 0}$ which coalesces quickly with high probability. 

In fact, the coupling that we will analyze is the optimal one-step coupling of the chains. Recall that this coupling is constructed as follows: given the current state $(X_{t-1}, Y_{t-1})$, we choose a uniformly random vertex $v \in V$ (in this case, we say that $v$ is updated at time $t$) and move to the state $(X_{t}, Y_{t})$ where
\begin{align*}
    &X_{t}(w) = X_{t-1}(w) \quad \forall w \neq v\\
    &Y_{t}(w) = Y_{t-1}(w) \quad \forall w\neq v\\
    &(X_{t}(v), Y_{t}(v)) \text{ is sampled from the optimal coupling of }(\mu_{\pi}[\on{value}(v)=\cdot \mid X_{t-1}^{-v}], \mu_{\pi}[\on{value}(v)=\cdot \mid Y_{t-1}^{-v}]).
\end{align*}

Hence, throughout the remainder of this section, $(X_t, Y_t)_{t\geq 0}$ will always denote the optimal one-step coupling of two copies of $Z_{t}$ starting at $X_0$ and $Y_0$ respectively. 

We partition time into blocks of size 
\[H = 100\kappa\cdot n,\]
where $\kappa \geq 2$ is the parameter appearing in (A2) of \cref{assumptions-1}.
By time block $K$, we mean the time interval $[HK, H(K+1))$. We will also need the following notation. 
Let $\mf{C} = \cal{C} \times \mathbb{Z}^{\geq 0}$. Recall that $G(\mc{C})$ denotes the graph whose vertex set consists of constraints $C \in \mc{C}$ and there is an edge between $C\neq C' \in \mc{C}$ if and only if $\on{vbl}(C)\cap \on{vbl}(C')\neq \emptyset$
Let 
$$U=\{(v,t) \in V \times \mathbb{Z}^{\geq 0}: v \textrm{ is updated at time }t\}.$$ 
Recall that this means that $v$ is the vertex chosen by the Glauber dynamics when the current states are $X_{t-1}$ and $Y_{t-1}$. Recall also that by the definition of the optimal one-step coupling, the same vertex is chosen to be updated at a given time $t$ in both chains. 
Let 
$$D = \{(v,t) \in V \times \mathbb{Z}^{\geq 0}: X_t(v)\ne Y_t(v)\}.$$ We call $D$ the set of {\it discrepancies}. 

Given a time interval $I$, we denote by $V^{0}(I)$ the set of variables $v\in V$ which are not updated in $I$ and by $V^{+}(I)$ the set of variables $v\in V$ which are updated at least $\kappa |I|/n$ times in $I$.

Finally, for $Z \in \prod_{v\in V}Q_{v}$, for $S\subseteq V$, and for $C \in \mc{C}$, we say that $S$ does not satisfy $C$ in $Z$ if
\[Z(v) = \boldsymbol{C}_{\pi}(v) \quad \forall v \in \on{vbl}(C) \cap S,\]
and that a partial assignment $Z'$ does not satisfy $C$ if 
\[Z'(v) = \boldsymbol{C}_{\pi}(v) \quad \forall v \in \on{vbl}(C)\text{ for which }Z'(v)\text{ is defined}.\]


\subsection{Discrepancy checks}




Our argument for bounding $\tau_{\on{couple}}$ will be based on showing that it is very unlikely for certain combinatorial structures, which we call \emph{minimal discrepancy checks}, to arise from the randomness of the choice of updates driving the Glauber dynamics. To this end, we begin by defining the notion of a {\it discrepancy check}. 
\begin{definition} \label{def:disc-check}
Let $v_0 \in V$ and $T_0$ be an integer in $[HK,H(K+1))$ (in particular, $T_0$ is in time block $K$). For $i\geq 1$, let $T_i = H(K-i)$. A {\it discrepancy check $\cal{D}$ starting at $(v_0,T_0)$} consists of a sequence of elements $(C_0, T_0), (C_1,T_1),\dots, (C_{K-1}, T_{K-1}) \in \mf{C}$, a sequence of elements $v_1,\dots,v_{K-1} \in V$, a collection of \emph{induced} oriented paths $\cal{P}_1,\dots,\cal{P}_{K-1}$ in $G(\cal{C})$, and a collection of Boolean variables $f_1,\dots,f_{K-1}$ satisfying the following properties. 
\begin{enumerate}[(D1)]
\item $(v_0,T_0) \in {D}$, $v_0 \in \on{vbl}(C_0)$ and $\on{vbl}(C_0)\setminus \{v_0\}$ does not satisfy $C_0$ in at least one of $X_{T_0-1}^{-v_0}$ and $Y_{T_0-1}^{-v_0}$. 
\item $f_1 = 1$, $v_1 \in \on{vbl}(C_1)$, and $\mc{P}_1$ is an induced path oriented from $C_0$ to $C_1$. Additionally, the following properties are satisfied.
\begin{itemize}
    \item $(v_1,T_1)\in D$ and either $X_{T_1}(v_1)$ or $Y_{T_1}(v_1)$ is equal to $(\boldsymbol{C_1})_{\pi}(v_1)$.
    \item For each constraint $C'$ in $\cal{P}_1 \setminus \{C_0\}$, there exists $T' \in [T_1, T_0)$ such that at least one of the following holds.
    \begin{itemize}
        \item The subset $\on{vbl}(C')$ 
        does not satisfy $C'$ in at least one of $X_{T'}$ and $Y_{T'}$. 
        \item There exists $v'\in \on{vbl}(C')$ such that 
        \begin{itemize}
        \item The subset $\on{vbl}(C')\setminus \{v'\}$
        does not satisfy $C'$ in at least one of $X_{T'}$ and $Y_{T'}$, and
        \item $v'$ is updated in $(T',T_0)$, and 
        \item The first update at time $t'>T'$ of $v'$ results in a discrepancy, and 
        \item There exists some $C'' \in \mc{C}$ such that either $X_{t'}(v')$ or $Y_{t'}(v')$ is equal to $\boldsymbol{C''}_\pi(v')$. 
        \end{itemize}
    \end{itemize}
\end{itemize}
We call the induced oriented path $\cal{P}_1$ from $C_0$ to $C_1$ the {\it $1$-leg} of the discrepancy check. 

\item For each $1 \le i \le K-2$, given $C_i$, $T_i$, $v_i$, and $\mc{P}_i$, if $v_i$ is not updated in $(T_{i+1}, T_i]$, then $f_{i+1} = 0$, $C_{i+1} = C_{i}$, $v_{i+1} = v_i$, and $\mc{P}_{i+1} = \{C_i\} = \{C_{i+1}\}$.

\item For each $1 \le i \le K-2$, given $C_i$, $T_i$, $v_i$, and $\mc{P}_i$, if $v_i$ is updated in $(T_{i+1}, T_i]$, then $f_{i+1} = 1$ and $\mc{P}_{i+1}$ is an induced path oriented from $\ol{C}_{i+1}$ to $C_{i+1}$. We require that the following properties are satisfied. 


\begin{itemize}
    \item $v_{i+1} \in \on{vbl}(C_{i+1})$, $(v_{i+1},T_{i+1})\in D$, and either $X_{T_{i+1}}(v_{i+1})$ or $Y_{T_{i+1}}(v_{i+1})$ is equal to $(\boldsymbol{C_{i+1}})_{\pi}(v_{i+1})$.
    \item $\ol{C}_{i+1}$ shares a variable with some constraint in $\cal{P}_i$. None of the constraints $C'\in \cal{P}_{i+1} \setminus \{\ol{C}_{i+1}\}$ share variables with any constraints in $\cal{P}_i$. 
    \item For each $C' \in \cal{P}_{i+1}$, there exists $T' \in [T_{i+1}, T_i)$ such that at least one of the following holds. 
    \begin{itemize}
        \item The subset $\on{vbl}(C')$
        does not satisfy $C'$ in at least one of $X_{T'}$ or $Y_{T'}$. 
        \item There exists $v'\in \on{vbl}(C')$ such that 
        \begin{itemize}
        \item The subset $\on{vbl}(C')\setminus \{v'\}$
        does not satisfy $C'$ in at least one of $X_{T'}$ and $Y_{T'}$, and
        \item $v'$ is updated in $(T',T_i)$, and 
        \item The first update at time $t'>T'$ of $v'$ results in a discrepancy, and 
        \item There exists some $C'' \in \mc{C}$ such that either $X_{t'}(v')$ or $Y_{t'}(v')$ is equal to $\boldsymbol{C''}_{\pi}(v')$.  
        \end{itemize}
    \end{itemize}
\end{itemize}
We call the induced oriented path $\cal{P}_{i+1}$ from $\ol{C}_{i+1}$ to $C_{i+1}$ the {\it $(i+1)$-leg} of the discrepancy check. 



\end{enumerate} 
\end{definition} 

For later use, we record the following simple lemma. 

\begin{lemma}
\label{lem:in_D}
Let $\mc{D}$ be a discrepancy check starting from $(v_0, T_0)$, where $T_0$ is in time block $K$. Then, for all $0\leq i \leq K-1$, $(v_i, T_i) \in D$.
\end{lemma}
\begin{proof}
By assumption, $(v_0, T_0) \in D$ and $(v_1, T_1) \in D$. Suppose for contradiction that there is some $2\leq i \leq K-1$ such that $(v_i, T_i)\notin D$ and let $i_*$ denote the smallest such index. Then, by the first bullet point of (D4), we cannot have $f_{i_*} = 1$. Therefore, we must have $f_{i_*} = 0$, in which case $v_{i_*} = v_{i_* - 1}$. But since $(v_{i_* - 1}, T_{i_* - 1}) \in D$ by the minimality of $i_*$ and since $v_{i_*} = v_{i_* - 1}$ is not updated in $(T_{i_*}, T_{i_*-1}]$ due to the condition $f_{i_*} = 0$, it follows that necessarily, $(v_{i_*}, T_{i_*}) \in D$, which contradicts the definition of $i_*$.
\end{proof}

\subsection{Constructing a discrepancy check} In this subsection, we show that whenever $(v_0,T_0)\in U \cap D$, there must exist a discrepancy check starting at $(v_0,T_0)$. 
\begin{proposition}\label{prop:disc-check-exists}
Let $(v_0,T_0) \in U\cap D$. Then there exists a discrepancy check $\cal{D}$ starting at $(v_0,T_0)$.
\end{proposition}

We divide the proof into a couple of lemmas. 

\begin{lemma} \label{lem:disc-path} 
Under the optimal one-step coupling of the Glauber dynamics for $\mu_{\pi}$, if $(v,t)\in U \cap D$, then there exists a path $C^1,C^2,\dots,C^k$ in $G(\cal{C})$ such that 
\begin{itemize}
\item $v\in \on{vbl}(C^1)$. 
\item Each $C^i$ is not satisfied in at least one of $X_{t-1}^{-v}$ and $Y_{t-1}^{-v}$. 
\item $\on{vbl}(C^k)$ contains some $u\ne v$ satisfying $X_{t-1}(u)\ne Y_{t-1}(u)$.  
\end{itemize}
\end{lemma}
\begin{remark}
Since $C^{k}$ is not satisfied by at least one of $X_{t-1}^{-v}$ and $Y_{t-1}^{-v}$ and since $u\neq v$, it follows in particular that either $X_{t-1}(u)$ or  $Y_{t-1}(u)$ is equal to $\boldsymbol{C^{k}}_{\pi}(u)$.
\end{remark}
\begin{proof}
Let $\mc{C}'$ denote those constraints which are not satisfied by at least one of $X_{t-1}^{-v}$ and $Y_{t-1}^{-v}$. Let $G'(\mc{C'})$ be graph on the vertex set $\mc{C'}$ induced by the graph $G(\mc{C})$. It is clear that the distributions $\mu_\pi[\on{value}(v) =  \cdot \mid X_{t-1}^{-v}]$ and $\mu_\pi[\on{value}(v) =  \cdot \mid Y_{t-1}^{-v}]$ depend only on the restrictions of $X_{t-1}^{-v}$ (respectively $Y_{t-1}^{-v}$) to the connected component of $v$ in $G'(\mc{C'})$. Therefore, if the connected component of $v$ in $G(\mc{C'})$ does not contain any variable $u\ne v$ for which $X_{t-1}(u)\ne Y_{t-1}(u)$, then under the optimal coupling of the Glauber dynamics, we must necessarily have $X_{t}(v) = Y_{t}(v)$, which contradicts $(v,t) \in U \cap D$. 
\end{proof}


The next lemma, which is more involved, shows how to inductively build a discrepancy check.

\begin{lemma}\label{lem:disc-leg}
For $1 \le \ell \le K-2$, given $(C_\ell, T_\ell) \in \mf{C}$, $v_\ell \in V$, and $\cal{P}_\ell$ satisfying the properties in \cref{def:disc-check}, there exist $(C_{\ell+1}, T_{\ell+1}) \in \mf{C}$, $f_{\ell+1}$, $v_{\ell+1} \in V  $ and $\cal{P}_{\ell+1}$ satisfying the properties of the $(\ell+1)$-leg in \cref{def:disc-check}. 
\end{lemma}
\begin{proof}
If $v_\ell$ is not updated in $(T_{\ell+1},T_\ell]$, then the choice $f_{\ell+1} = 0$, $v_{\ell+1} = v_{\ell}$, $C_{\ell+1} = C_{\ell}$ and $\mc{P}_{\ell+1} = \{c_{\ell+1}\}$ satisfies (D3) and we are done.\\ 

Otherwise, $v_{\ell}$ is updated in $(T_{\ell+1}, T_{\ell}]$. We claim that there exists an induced oriented path of constraints $\cal{P} = (C^0,\dots, C^k)$ with the following properties:
\begin{enumerate}[(Q1)]
    \item $v_{\ell} \in \on{vbl}(C^0)$.
    \item There exists some $v_{\ell+1} \in \on{vbl}(C^k)$ such that $(v_{\ell+1}, T_{\ell+1}) \in D$ and either $X_{T_{\ell+1}}(v_{\ell+1})$ or $Y_{T_{\ell+1}}(v_{\ell+1})$ is equal to $(\boldsymbol{C^k})_{\pi}(v_{\ell+1})$.
    \item For each constraint $C \in \mc{P}$, there exists some $T' \in [T_{\ell+1}, T_{\ell})$ such that at least one of the following holds. 
    \begin{itemize}
    \item The subset $\on{vbl}(C)$ does not satisfy $C$ in at least one of $X_{T'}$ and $Y_{T'}$.
    \item There exists $v' \in \on{vbl}(C)$ such that 
        \begin{itemize}
        \item The subset $\on{vbl}(C)\setminus \{v'\}$ does not satisfy $C$ in at least one of $X_{T'}$ and $Y_{T'}$, and
        \item $v'$ is updated in $(T',T_\ell)$, and 
        \item The first update at time $t'>T'$ of $v'$ results in a discrepancy, and 
        \item There exists some $C'' \in \mc{C}$ such that either $X_{t'}(v')$ or $Y_{t'}(v')$ is equal to $\boldsymbol{C''}_\pi(v')$. 
        \end{itemize}
    \end{itemize}
\end{enumerate}
We now show that such a path exists. This is the main step in the proof.\\ 

Let $T'_\ell$ be the last update of $v_\ell$ in the interval $(T_{\ell+1}, T_{\ell}]$. Since $(v_{\ell}, T_{\ell}) \in D$ (\cref{lem:in_D}), it follows that $(v_{\ell}, T'_{\ell}) \in D$ and hence in $(v_{\ell}, T'_{\ell}) \in U\cap D$. 
Therefore, letting $T^0 = T'_{\ell}-1$, it follows by \cref{lem:disc-path} that there exists a path  $\cal{P}^{0}=(C^{0}_1,\dots,C^{0}_{s_0})$ in $G(\mc{C})$ such that $v_{\ell} \in \on{vbl}(C^{0}_1)$, each $C^{0}_j$ is not satisfied by at least one of $X_{T^0}^{-v_{\ell}}$ and $Y_{T^0}^{-v_{\ell}}$, and there exists $v^0 \neq v_{\ell}$ with $v^0 \in \on{vbl}(C^0_{s_0})$ and $(v^0, T^0) \in D$. In particular, either $X_{T^0}(v^0)$ or $Y_{T^0}(v^0)$ is equal to $(\boldsymbol{C^{0}_{s_0}})_{\pi}(v^0)$. 
By choosing such a path of minimum length, we may assume that $\mc{P}^0$ is an induced oriented path in $G(\mc{C})$. We have the following two cases.\\ 

\textbf{Case I: }$T^0 = H(K-\ell-1) = T_{\ell+1}$. Then, $\mc{P}_{\ell+1} = \mc{P}^{0}$ satisfies (Q1), (Q2), and (Q3). Indeed, (Q1) and (Q2) (with $v_{\ell+1} = v^0$) are clear. 
Moreover, the constraints $C^{0}_{j}$ for $j\geq 2$ satisfy the first bullet point of (Q3). Finally, the constraint $C^{0}_{1}$ satisfies the second bullet point of (Q3) with $v' = v_{\ell}$, $T' = T_{\ell+1} = T^{0}$, $t' = T^{0} + 1 = T'_{\ell}$, and $C'' = C_{\ell}$. Indeed, we know by the first bullet point of (D4) that either $X_{T_{\ell}}(v_{\ell})$ or $Y_{T_{\ell}}(v_{\ell})$ is equal to $(\boldsymbol{C_{\ell}})_{\pi}(v_{\ell})$ and since $t' = T'_{\ell}$ is the time of the last update to $v_{\ell}$ before $T_{\ell}$, it must be the case that either $X_{t'}(v')$ or $Y_{t'}(v')$ is equal to $(\boldsymbol{C_\ell})_{\pi}(v')$.\\


\textbf{Case II: } $T^0 > H(K-\ell-1) = T_{\ell+1}$. By induction, suppose that for $j \geq 0$, we have an induced oriented path $\mc{P}^{j}$ in $G(\mc{C})$ with $\mc{P}^{j} = (C^{j}_1,\dots, C^{j}_{s_j})$, a variable $v^{j} \in V$, and $T_{\ell} \geq T^{j} > H(K-\ell-1)$ with the following properties:
\begin{enumerate}[(R1)]
    \item $v_{\ell} \in \on{vbl}(C^{j}_{1})$.
    \item $v^j \in \on{vbl}(C^j_{s_j})$, $(v^j, T^j) \in D$, and either $X_{T^{j}}(v^j)$ or $Y_{T^j}(v^j)$ is equal to $(\boldsymbol{C^j_{s_j}})_{\pi}(v^j)$. 
    \item For each constraint $C \in \mc{P}^{j}$, there exists some $T' \in [T_{\ell+1} ,T_{\ell})$ such that at least one of the following holds.
    \begin{itemize}
        \item The subset $\on{vbl}(C)$ does not satisfy $C'$ in at least one of $X_{T'}$ or $Y_{T'}$.
        \item There exists $v' \in \on{vbl}(C)$ such that
        \begin{itemize}
        \item The subset $\on{vbl}(C)\setminus \{v'\}$ does not satisfy $C$ in at least one of $X_{T'}$ and $Y_{T'}$, and
        \item $v'$ is updated in $(T',T_\ell)$, and \item The first update at time $t'>T'$ of $v'$ results in a discrepancy, and 
        \item There exists some $C'' \in \mc{C}$ such that either $X_{t'}(v')$ or $Y_{t'}(v')$ is equal to $\boldsymbol{C''}_\pi(v')$. 
        \end{itemize}
    \end{itemize}
\end{enumerate}

Let $\ol{T}^j$ be the last update of $v^j$ with $\ol{T}^j \le T^j$. We have two cases.\\

\textbf{Case 1}: If $\ol{T}^j \le T_{\ell+1}$, then the path $\cal{P}^j$ satisfies the required properties (Q1), (Q2), (Q3). Indeed, (R1) implies (Q1), (R3) implies (Q3), and (R2) implies (Q2) since $\ol{T}^{j} \leq T_{\ell+1}$ implies that $X_{T^j}(v^j) = X_{T_{\ell+1}}(v^j)$ and $Y_{T^j}(v^j) = Y_{T_{\ell+1}}(v^j)$.\\  

\textbf{Case 2}: $\ol{T}^j > T_{\ell+1}$. Since $(v^{j}, T^{j}) \in D$ by (R2), we must have $(v^{j}, \ol{T}^{j}) \in U \cap D$. Therefore, by \cref{lem:disc-path}, there exists an induced path $\cal{P}'^j = ({C'}^{j}_{1},\dots, {C'}^{j}_{s'_j})$ with $v^j \in \on{vbl}({C'}^{j}_1)$, each ${C'}^{j}_{i}$ is not satisfied by at least one of $X_{T^{j+1}}^{-v^{j}}$ and $Y_{T^{j+1}}^{-v^{j}}$, where $T^{j+1} = \ol{T}^{j}-1$, and there exists $v^{j+1} \neq v^{j}$ with $v^{j+1} \in \on{vbl}({C'}^{j}_{s'_j})$ and $(v^{j+1},T^{j+1})\in D$. Concatenating $\cal{P}^j$ with $\cal{P}'^j$ gives an oriented path from $C^{j}_{0}$ to ${C'}^{j}_{s'_j}$. By taking a sub-path between these endpoints which is an induced oriented path in $G(\mc{C})$, we get $\mc{P}^{j+1}$, which satisfies properties (R1), (R2), and (R3) with $j+1$. Indeed, (R1) follows from the assumption (R1) for $\mc{P}^j$, (R2) follows from \cref{lem:disc-path} and the remark following it whereas (R3) follows from \cref{lem:disc-path} and the assumptions (R2) and (R3) for $\mc{P}_j$. 

Note that, by construction, we have $T_{\ell+1} \leq T^{j+1}<T^{j}$. If $T^{j+1} = T_{\ell+1}$, then, as before, $\cal{P}^{j+1}$ satisfies the properties (Q1), (Q2), and (Q3). If $T^{j+1} > T_{\ell+1}$, then we repeat, noting that the process is guaranteed to terminate in finitely many steps since the sequence $(T^j)_{j\geq 0}$ is strictly decreasing before termination. This completes the proof of our claim about the existence of an induced oriented path satisfying (Q1), (Q2), and (Q3).\\

Now, let $\mc{P} = (C^{0},\dots, C^{k})$ be an induced oriented path in $G(\mc{C})$ satisfying (Q1), (Q2), and (Q3). Let $\ol{C}_{\ell+1}$ be the last (according to the orientation) constraint in $\mc{P}$ which shares a variable with any constraint in $\mc{P}_{\ell}$. Let $\mc{P}_{\ell+1}$ denote the part of $\mc{P}$ starting from $\ol{C}_{\ell+1}$. Also, denote the last constraint in $\mc{P}_{\ell+1}$ by $C_{\ell+1}$ and let $\boldsymbol{C}_{\ell+1} = (C_{\ell+1}, T_{\ell+1})$. We claim that $\mc{P}_{\ell+1}$ satisfies the properties of the $(\ell+1)$-leg of the discrepancy check. Indeed, the first bullet point in (D4) follows from (Q2), the second bullet point of (D4) follows from the construction of $\mc{P}_{\ell+1}$, and the third bullet point of (D4) follows from (Q3). \qedhere

\end{proof}

Given the preceding two lemmas, the proof of \cref{prop:disc-check-exists} follows easily.

\begin{proof}[Proof of \cref{prop:disc-check-exists}]
Since $(v_0, T_0) \in U \cap D$, it follows by \cref{lem:disc-path} that there exists an induced oriented path $\mc{P}^{0} = (C_{1}^{0},\dots, C_{s_0}^{0})$ in $G(\mc{C})$ such that $v_0 \in \on{vbl}(C_{1}^{0})$, each $C_{j}^{0}$ is not satisfied by at least one of $X_{T_0 - 1}^{-v_0}, Y_{T_0 - 1}^{-v_0}$, and there exists $v^{0} \neq v_{0}$ with $v^{0} \in \on{vbl}(C^{0}_{s_0})$ and $(v^{0}, T_0 - 1) \in D$. Then, by the same argument as in the proof of \cref{lem:disc-leg} (the only difference is that we slightly weaken the condition (R3) and require it only for $C \in \mc{P}^{j} \setminus \{C^{j}_{1}\}$), we can show that there exists an induced oriented path of constraints $\mc{P}_{1} = (C^{0},\dots, C^{k})$ with $C^{0} = C_0$ and $C^{k} = C_{1}$ and such that the following properties hold. 
\begin{enumerate}[(Q'1)]
    \item $v_0 \in \on{vbl}(C_0)$.
    \item There exists some $v_1 \in \on{vbl}(C_1)$ such that $(v_1, T_1) \in D$ and either $X_{T_1}(v_1)$ or $Y_{T_1}(v_1)$ is equal to $(\boldsymbol{C_1})_{\pi}(v_1)$.
    \item For each constraint $C \in \mc{P}_1 \setminus \{C_0\}$, there exists some $T' \in [T_{1}, T_0)$ such that at least one of the following holds.
    \begin{itemize}
    \item The subset $\on{vbl}(C)$ does not satisfy $C$ in at least one of $X_{T'}$ and $Y_{T'}$.
    \item There exists $v' \in \on{vbl}(C)$ such that 
        \begin{itemize}
        \item The subset $\on{vbl}(C)\setminus \{v'\}$ does not satisfy $C$ in at least one of $X_{T'}$ and $Y_{T'}$, and
        \item $v'$ is updated in $(T',T_\ell)$, and 
        \item The first update at time $t'>T'$ of $v'$ results in a discrepancy, and 
        \item There exists some $C'' \in \mc{C}$ such that either $X_{t'}(v')$ or $Y_{t'}(v')$ is equal to $\boldsymbol{C''}_\pi(v')$. 
        \end{itemize}
    \end{itemize}
\end{enumerate}
For such a path, note that $\boldsymbol{C_1} = (C_1, T_1) \in \mf{C}$, $f_1 = 1$, $v_1 \in V$, and $\mc{P}_1$ satisfy the properties in \cref{def:disc-check}. Now, a direct (repeated) application of \cref{lem:disc-leg} gives a discrepancy check starting at $(v_0, T_0)$. 
\end{proof}

\subsection{Minimal discrepancy checks}
In order to carry out the union bound argument later (both to control the size of the union as well as to control the probabilities of individual events in the union), it will be convenient to focus on {\it minimal discrepancy checks}. 
\begin{definition}\label{def:minimal-disc-check}
Let $v_0 \in V$ and $T_0$ an integer in $[HK,H(K+1))$ (in particular, $T_0$ is in time block $K$).  For $1\leq i \leq K-1$, let $T_{i} = H(K-i)$. A {\it minimal discrepancy check $\cal{M}$ starting at $(v_0,T_0)$} consists of a sequence of induced oriented paths $\mathsf{P}_1,\dots,\mathsf{P}_{K-1}$ in $G(\cal{C})$ and a collection of Boolean variables $f_1 = 1,f_2,\dots,f_{K-1}$ 
such that the following properties are satisfied. 
\begin{enumerate}[(M1)]
\item $(v_0, T_0) \in D$.
\item The first constraint of $\mathsf{P}_1$, which we denote by $C_0$, satisfies $v_0 \in \on{vbl}(C_0)$. Moreover,  
$\on{vbl}(C_0)\setminus \{v_0\}$ does not satisfy $C_0$ in at least one of $X_{T_0-1}^{-v_0}$ and $Y_{T_0-1}^{-v_0}$. 
\item For $i\ge 0$ satisfying $f_{i+1} = 0$, let $j(i+1) = \max\{j: j\leq i+1, f_j = 1\}$. Then, there exists some $v_{j(i+1)} \in \on{vbl}(C_{j(i+1)})$, where $C_{j(i+1)}$ is the last constraint in $\mathsf{P_{j(i+1)}}$, such that $(v_{j(i+1)}, T_i) \in D$ and $v_{j(i+1)}$ is not updated in $(T_{i+1}, T_i)$. Moreover, $\mathsf{P}_{i+1} = \{C_i\}$. 
\item For $i\geq 0$ satisfying $f_{i+1} = 1$, the following properties hold. 
    \begin{itemize}
        \item The last constraint of $\mathsf{P_i}$ and the first constraint of $\mathsf{P_{i+1}}$ have non-empty intersection. Any other pair of constraints in $\mathsf{P}_i$ and $\mathsf{P}_{i+1}$ are disjoint. 
        \item For any $C' \in \mathsf{P}_1 \setminus \{C_0\}$ (in case $i = 0$) and for any $C' \in \mathsf{P}_{i+1}$ (in case $i\geq 1$), let
        \[\on{vbl}^0(C') := \on{vbl}(C')\cap V^0((T_{i+2},T_{i+1})),\] 
        and 
        \[\on{vbl}^+(C') := \on{vbl}(C')\cap V^{+}((T_{i+1},T_{i})).\] 
        Then, there exists some $T' \in [T_{i+1}, T_i)$ such that at least one of the following holds.  
            \begin{itemize}
                \item The subset 
                \[\ol{\on{vbl}}(C')\setminus (\on{vbl}^0(C') \cup \on{vbl}^+(C'))\] 
                does not satisfy $C'$ in at least one of $X_{T'}$ or $Y_{T'}$. 
                \item There exists $v' \in \ol{\on{vbl}}(C') \setminus (\on{vbl}^{0}(C') \cup \on{vbl}(C'))$ such that 
                \begin{itemize}
                \item The subset 
                \[\ol{\on{vbl}}(C') \setminus (\{v'\} \cup \on{vbl}^0(C') \cup \on{vbl}^+(C'))\] 
                does not satisfy $C'$ in at least one of $X_{T'}$ and $Y_{T'}$, and
                \item $v'$ is updated in $(T',T_i)$, and 
                \item The first update at time $t'>T'$ of $v'$ results in a discrepancy, and 
                \item There exists some $C'' \in \mc{C}$ such that either $X_{t'}(v')$ or $Y_{t'}(v')$ is equal to $\boldsymbol{C''}_{\pi}(v')$. 
                \end{itemize}
            \end{itemize}
    \end{itemize}
\end{enumerate}
For a minimal discrepancy check $\mc{M}$, we denote the number of constraints in $\mathsf{P}_i$ by $r_i$ and refer to it as the \emph{length} of leg $i$.
\end{definition}

The next lemma shows how to modify a discrepancy check in order to obtain a minimal discrepancy check.
\begin{lemma}
\label{lem:minimal}
Suppose there exists a discrepancy check $\mc{D}$ starting at $(v_0, T_0)$. Then, there exists a minimal discrepancy check $\mc{M}$ starting at $(v_0, T_0)$.
\end{lemma}
\begin{proof}
Let $\cal{D}$ be a discrepancy check starting at $(v_0,T_0)$. For each $i \ge 1$, let $\tilde{C}_i$ denote the first (according to the orientation) constraint in $\mc{P}_i$ for which $\on{vbl}(\tilde{C}_i) \cap \on{vbl}(\ol{C}_{i+1}) \neq \emptyset$. Let $\mathsf{P_i}$ denote the part of $\mc{P}_i$ from the starting point until and including $\tilde{C}_i$. We claim that the paths $\mathsf{P}_1,\dots, \mathsf{P}_{K-1}$ and the Boolean variables $f_1,\dots, f_{K-1}$ from $\mc{D}$ constitute a minimal discrepancy check starting at $(v_0, T_0)$, where we use the variables $v_i$ from $\mc{D}$ for each $i$ satisfying $f_{i+1} = 0$ in order to check condition (M3).

Indeed, (M1) and (M2) follow from (D1). (M3) follows from (D3) and \cref{lem:in_D}.
The first bullet point of (M4) follows from the the second bullet point of (D4), the construction of $\mathsf{P}_i$s, and the fact that the $\mc{P}_j$s are induced paths. The second bullet point of (M4) follows from the second bullet point of (D2) and the third bullet point of (D4).  
\end{proof}

Combining this lemma with \cref{prop:disc-check-exists}, we have the following. 
\begin{proposition}\label{pro:min-disc-check-exists}
Let $(v_0,T_0) \in U\cap D$. Then there exists a minimal discrepancy check $\cal{M}$ starting at $(v_0,T_0)$.
\end{proposition}

To prepare for the next subsection, we introduce some more notation. To every minimal discrepancy check $\mc{M}$, we associate a graph $G(\mc{M}) = (V(\mc{M}), E(\mc{M}))$ defined as follows. The vertex set $V(\mc{M})$ consists of pairs $(C,i)$ where $1\leq i \leq K-1$ and $C \in \mathsf{P}_i$. The vertices $(C,i)$ and $(C',i)$ are connected to each other if and only if $C$ and $C'$ are adjacent in $\mathsf{P}_i$. Moreover, if $C$ is the last constraint in $\mathsf{P}_i$ and $C'$ is the first constraint in $\mathsf{P}_{i+1}$, then $(C,i)$ and $(C', i+1)$ are adjacent. Given the Boolean variables $f_1,\dots, f_{K-1}$ of $\mc{M}$, we can find some $1\leq \ell \leq K -1$ and disjoint intervals $L_1,\dots, L_\ell \subset \{1,\dots, K-1\}$ with the following properties. 
\begin{itemize}
    \item For all $a < b$, $\max{L_a} < \min{L_b}$.
    \item For every $1\leq s \leq \ell$ and for every $i \in L_s$, $f_i = 1$.
    \item For every $i$ such that $f_i = 1$, there exists some $1\leq s \leq \ell$ such that $i \in L_s$.
\end{itemize}
Given $L_1,\dots, L_\ell$, we define oriented induced paths $\wh{P}_1,\dots, \wh{P}_{\ell}$ in $G(\mc{M})$ where $\wh{P}_1$ contains all the points $\{(C,a): a \in L_1, C \in \mathsf{P}_a\}\setminus (C_0,1)$ and for $2\leq j \leq \ell$, $\wh{P}_j$ contains all the points $\{(C,a): a \in L_j, C \in \mathsf{P}_a\}$. For $1\leq j \leq \ell$, let 
\[\wh{r}_j = \sum_{a \in L_j}r_a.\]
We say that $\ell$ is the \emph{effective parameter} of the minimal discrepancy check and $\wh{r}_1,\dots, \wh{r}_{\ell}$ are its \emph{effective leg lengths}. 

Finally, for each $1 \leq j \leq \ell$, by taking every other point $\wh{P}_j$ starting with the very last point, we obtain an independent set $I_1(\mc{M})$ in $G(\mc{M})$ such that 
\[|I_1(\mc{M})| \geq \sum_{j=1}^{\ell}\wh{r}_j/2.\]
We have two cases.
\begin{enumerate}[({I}nd1)]
    \item $\sum_{j=1}^{\ell}\wh{r}_j \geq \sum_{i=1}^{K-1}r_i/7$. In this case, we define $I_0(\mc{M}) := \emptyset$ and $I(\mc{M}) := I_1(\mc{M})$.
    \item $\sum_{j=1}^{\ell}\wh{r}_j < \sum_{i=1}^{K-1}r_i/7$. In this case, observe that we can find an independent set $I_0(\mc{M})$ in $G(\mc{M})$ such that every element $(C,i) \in I_0(\mc{M})$ satisfies $f_i = 0$, $|I_0(\mc{M})| \geq \sum_{i=1}^{K-1}r_i/7$, and $I(\mc{M}) := I_0(\mc{M}) \cup I_1 (\mc{M})$ is also an independent set in $G(\mc{M})$.
\end{enumerate}
\subsection{Probability of a minimal discrepancy check}
In the previous subsection, we showed that if $(v_0, T_0) \in U \cap D$, then there must exist a minimal discrepancy check starting at $(v_0, T_0)$. In this subsection, we will bound the probability (under the randomness driving the Glauber dynamics) of seeing a minimal discrepancy check with given leg lengths. This will then be combined with a union bound argument in the next subsection.\\

It will be convenient to use the following description of the Glauber dynamics for $\mu_{\pi}$. To each vertex $v$ and each time $t$, we associate an independent uniform random variable $U(v,t) \sim \on{Unif}[0,1]$. At time $t$ (recall that this means that the current states are $X_{t-1}, Y_{t-1}$), we choose a variable $v$ to update from the uniform distribution on $V$ and choose $X_t(v)$ and $Y_t(v)$ according to the one-step maximal coupling of the corresponding conditional marginal distributions, with the realisation of $(X_t(v), Y_t(v))$ determined using $U(v,t)$ in the natural manner. With this notation, observe the following.
\begin{enumerate}[(O1)]
    \item For each $t, X_{t-1}, Y_{t-1}, v$, there exist subsets $I_d(v, X_{t-1}, Y_{t-1}) \subseteq [0,1]$ of measure at most $2q$ such that if $(v,t)\in U\cap D$, then $U(v,t)\in I_d(v)$. Here, as was defined in \cref{sec:projection},
    \[
    q = \max_{v\in V, Y\in \prod_{v\in V}Q_v} d_{\on{TV}}(\mb{P}_{\pi}[\on{value}(v) = \cdot], \mu_\pi[\on{value}(v) =  \cdot \mid Y^{-v}]).
    \]
    \item Suppose $e\cdot b \cdot \Delta \leq 1$. Then, by \cref{lem:marginal}, for each $v, C$, there exist subsets $I_s(v,C) \subseteq [0,1]$ of measure at most \[(1-3b)^{-\Delta}\p_{\pi}[\on{value}(v)=\boldsymbol{C}_\pi(v)]\] such that if $v$ does not satisfy $C$ in at least one of $X_t$ and $Y_t$ and if $t'$ is the last update of $v$ before time $t$, then $U(v,t') \in I_s(v,C)$. Here, recall that
    \[b = \max_{C\in \mc{C}}b(C),\]
    where
    \[b(C) = \prod_{u \in \on{vbl}(C)}|\pi_{u}^{-1}(\boldsymbol{C}_{\pi}(u))|^{-1}.\]
\end{enumerate}

Let $\mathbf{P}$ denote the probability measure corresponding to the randomness of the Glauber dynamics i.e. the choice of vertex to update at time $t$ and the i.i.d.~random variables $U(v,t) \sim \on{Unif}[0,1]$. The main result of this subsection is the following. 
\begin{proposition} \label{prop:badprob-check}
Let $\Phi$ be an atomic CSP and let $\pi$ be an admissible projection scheme. 
Let $\cal{M}$ be a minimal discrepancy check starting from $(v_0,T_0)$, where $T_0\ge HK$ ($H = 100\kappa n$), and with leg lengths $(r_1,\dots,r_{K-1})$. Further, let $\ell$ be the effective parameter of $\mc{M}$ and let its effective leg lengths be $(\wh{r}_1,\dots, \wh{r}_{\ell})$. 
Then,
\[
\mathbf{P}(\cal{M}) \le (3000\Delta)^{-\sum_{i=1}^{\ell} \hat{r}_i}\cdot (3000\Delta)^{-2|I_0(\mc{M})|} 
\]
\end{proposition}


\begin{proof}
Recall the definition of the graph $G(\mc{M})$ associated to $\mc{M}$ and the independent set $I(\mc{M}) = I_0(\mc{M}) \cup I_1(\mc{M})$ in this graph. Recall also that $I(\mc{M})$ does not contain $(C_0, 1)$.
For $1\leq i \leq K-1$, let
\[\mf{C}_i = \{(C,i) \in I(\mc{M})\}.\]
Also, let
\begin{align*}
    \mc{V}^{0}(C,i) &= V^{0}((T_{i+1}, T_{i}))\cap \ol{\on{vbl}}(C), \\
    \mc{V}^{+}(C,i) &= V^{+}((T_{i}, T_{i-1}))\cap \ol{\on{vbl}}(C),
\end{align*}
and let
\begin{align*}
    \mf{V}^0_i &= \cup_{(C,i) \in \mf{C}_i}\mc{V}^{0}(C,i),\\
    \mf{V}^+_i &= \cup_{(C,i) \in \mf{C}_i}\mc{V}^{+}(C,i).
\end{align*}
Note that for $1 \leq i \leq K-1$, the sets $\mf{V}^+_{i}$ and $\mf{V}^0_{i-1}$ are completely determined by the choice of vertices selected to be updated between times $(T_{i}, T_{i-1}) =: \on{Int}_{i-1}$. We say that $\on{Int}_{i-1}$ is \emph{exceptional} if 
\[|\mf{V}^+_{i}| + |\mf{V}^0_{i-1}| \geq \frac{n}{2}.\]
It follows from a straightforward application of the Chernoff bound that for any $1 \leq i \leq K-1$,
\begin{align}
    \label{eqn:exceptional}
    \mathbf{P}[\on{Int}_{i-1} \text{ is exceptional}] \leq 2\exp(-H/12) \leq \exp(-8\kappa n).
\end{align}
Moreover, for disjoint subsets $V^0$ and $V^+$ of $V$ such that $|V^0| + |V^+| \leq n/2$, we have
\begin{align}
\label{eqn:unexceptional}
    \mathbf{P}[\mf{V}^{0}_{i-1} = V^0, \mf{V}^+_{i} = V^+]
    &= \mathbf{P}[\mf{V}^+_{i} = V^+ \mid \mf{V}^0_{i-1} = V^0]\mathbf{P}[\mf{V}^0_{i-1} = V^0]\\ \nonumber 
    &\leq \exp\left(-\frac{H|V^+|}{6n}\right)\cdot\left(1-\frac{|V^0|}{n}\right)^{H}\\ \nonumber 
    &\leq e^{-10\kappa\cdot (|V^0| + |V^+|)},
\end{align}
where in the second line, we have used that conditioned on $\mf{V}_{i-1}^0 = V^0$, the vertex to be updated at each step is chosen uniformly from a set of size at least $n/2$, and further, membership in $\mf{V}_{i}^+$ for different vertices is negatively dependent.\\ 

Let $\mf{I} \subset \{1,\dots, K-1\}$ denote the (random) set of $i$ such that $\on{Int}_{i-1}$ is exceptional. By the law of total probability, it suffices to show that
\[\mathbf{P}[\mc{M} \cap \{\mf{I} = \hat{J}\}] \leq (3000\Delta)^{-\sum_{i=1}^{\ell}\wh{r}_i}\cdot (3000\Delta)^{-2|I_0(\mc{M})|}.\] 
for every subset $\hat{J}$ of $\{1,\dots, K-1\}$. Therefore, for the remainder of the proof, we fix such a choice of $\hat{J}$ and let
\[\hat{I} = \{i \in \{1,\dots, K-1\}: \{i-1, i, i+1\} \notin \hat{J}\}.\]
Also, for each $i \in \hat{I}$, let $\mf{T}_{i-1}$ denote the sigma-algebra generated by the random variables which record the choice of vertex to update for times between $[T_i, T_{i-1})$. Note that $\mf{V}^{0}_i$ and $\mf{V}^+_i$ are measurable with respect to $\mf{T}_{i}$ and $\mf{T}_{i-1}$ respectively. We denote the realizations of these random sets, given the relevant sigma algebras, by $V^0_i(\mf{T}_{i})$ and $V^+_i(\mf{T}_{i-1})$ respectively. Note that, by the definition of $\hat{I}$, it is necessarily the case that $|V^+_i(\mf{T}_{i-1})| + |V^0_{i-1}(\mf{T}_{i-1})|\leq n/2$. We also let
\[\hat{I}_1 := \{i \in \hat{I}: f_i = 1\},\]
and
\[\hat{I}_0 := \{i \in \hat{I}: f_i = 0, \exists C \in \mc{C} \text{ such that }(C,i) \in I_0(\mc{M})\}.\]
Observe that, by construction of the set $I_0(\mc{M})$, we have that for any $i \in \hat{I}_0$ and any $j \in \hat{I}_1$,
\begin{equation}
\label{eqn:disjoint-interval}
\on{Int}_{i-1} \cap (\on{Int}_{j-1} \cup \on{Int}_j) = \emptyset.
\end{equation}

\medskip

\medskip 

Now, consider $i \in \hat{I}_1$. Conditioning on $\mf{T}_{i-1}, \mf{T}_{i}$ fixes $V^+_i = V^+_i(\mf{T}_{i-1})$ and $V^0_i = V^0_{i}(\mf{T}_{i})$. Moreover, conditioning on $\mf{T}_i$ fixes, for each $\wh{C} = (C,i) \in \mf{C}_i$, the set $\mc{T}(\wh{C}) \subseteq [T_i, T_{i-1})$ consisting of $T_i$ and all the update times of each variable $v \in \ol{\on{vbl}}(C) \setminus V_i^+$. Observe that 
\[|\mc{T}(\wh{C})| \leq 2\kappa H/n \cdot |\ol{\on{vbl}}(C)|\] and that the following holds. 

\begin{itemize}
    \item If there exists $T' \in [T_i,T_{i-1})$ such that $C$ is not satisfied in at least one of $X_{T'}$ and $Y_{T'}$ by $\ol{\on{vbl}}(C) \setminus ({V}^{0}_i \cup {V}^{+}_i)$, then there exists some $t\in \cal{T}(\wh{C})$ such that in the last update $t_v\le t$ of each variable $v\in \ol{\on{vbl}}(C)\setminus (V^0_i \cup V^+_i)$ before time $t$, we necessarily have $U(v,t_v) \in I_s(v,C)$. Indeed, $t$ can be taken to simply be the maximum of $T_{i}$ and the last time before (and including) $T'$ that any variable in $\ol{\on{vbl}}(C)\setminus (V^0_i \cup V^+_i)$ is updated. Call this event $\cal{E}_0(\wh{C})$.
        
        Since conditioned on $\mf{T}_i$ and $\mf{T}_{i-1}$, $t_{v}$ is determined by $t$ for each $v \in \ol{\on{vbl}}(C) \setminus V^0_i$, we have
        \begin{align*}
            \mathbf{P}[\mc{E}_{0}(\wh{C}) \mid \mf{T}_i, \mf{T}_{i-1}]
            & \leq |\cal{T}(\wh{C})|\prod_{v\in \ol{\on{vbl}}(C) \setminus ({V}^0_i \cup {V}^+_i)} \mathbf{P}[U(v,t_v)\in I_s(v,C) \mid \mf{T}_i, \mf{T}_{i-1}, t]\\
            &\leq 200 |\ol{\on{vbl}}(C)|\kappa^{2}\cdot \prod_{v\in \ol{\on{vbl}}(C) \setminus (V^0_i\cup V^+_i)} (1-3b)^{-\Delta} \p_{\pi}[\on{value}(v)=\boldsymbol{C}_\pi(v)].
        \end{align*}
   \item If there exists $T' \in [T_i,T_{i-1})$ and $v' \in \ol{\on{vbl}}(C) \setminus (V^0_i\cup V^+_i)$ such that
    \begin{itemize}
        \item  $C$ is not satisfied in $X_{T'}$ or $Y_{T'}$ by $\ol{\on{vbl}}(C)\setminus \{v'\}$, and
        \item the first update $t' > T'$ of $v'$ results in a discrepancy, and
        \item there exists some $C' \in \mc{C}$ with $v' \in \on{vbl}(C')$ such that either $X_{t'}(v')$ or $Y_{t'}(v')$ is equal to $\boldsymbol{C'}_{\pi}(v')$,
    \end{itemize}
    then there exists some $t \in \mc{T}(\wh{C})$ (again, $t$ can be taken to be the maximum of $T_i$ and the last time before (and including) $T'$ that any variable in $\ol{\on{vbl}}(C) \setminus (V^0_i\cup V^+_i)$ is updated) such that
    \begin{itemize}
        \item in the last update $t_v \leq t$ of each variable $v\in \ol{\on{vbl}}(C) \setminus (\{v'\} \cup V^0_i\cup V^+_i)$, we necessarily have $U(v,t_v) \in I_s(v,C)$, and
        \item $U(v',t') \in I_{d}(v') \cap I_{s}(v', C')$.
    \end{itemize}
  Call this event $\cal{E}_1(\wh{C})$.\\ 
  
  Let $\mf{T}_{\leq i}$ denote the sigma algebra generated by $\mf{T}_1,\dots, \mf{T}_i$.
 Note that, conditioned on $\mf{T}_{\leq i}$, $t \in \mc{T}(\wh{C})$ determines $t_{v}$ (for all $v \in \ol{\on{vbl}}(C)\setminus V_i^0$), $v'$ and $t'$. Therefore, letting
  \[V'(\wh{C}) := \ol{\on{vbl}}(C)\setminus (V^0_i\cup V^+_i \cup \{v'\}),\]
  we have
%
\begin{align*}
    &\mathbf{P}[\mc{E}_{1}(\wh{C}) \mid \mf{T}_{\leq i}] \\
    & \leq |\cal{T}(\wh{C})|\cdot \mathbf{P}[U(v',t') \in I_d(v') \cap I_s(v',C') \mid \mf{T}_{\leq i}, t]\cdot \prod_{v\in V'(\wh{C})} \mathbf{P}[U(v,t_v)\in I_s(v,C) \mid \mf{T}_{\leq i}, t] \\
    &\le 200|\ol{\on{vbl}}(C)|\kappa^2\cdot \min\left(q, \Delta(1-3b)^{-\Delta}\max_{C': v'\in \on{vbl}(C')}\p_{\pi}[
    \on{value}(v')= \boldsymbol{C'}_\pi(v')]\right)\\ 
    &\quad \quad \quad \cdot \prod_{v \in V'(\wh{C})} (1-3b)^{-\Delta} \p_{\pi}[\on{value}(v)=\boldsymbol{C}_\pi(v)].
    \end{align*}
\end{itemize}
Here, the factor of $\Delta$ in the second term in the parentheses is to account for the choice of $C'$ given $v'$. 

Note that by the definition of a minimal discrepancy check, if $f_{i} = 1$ (in particular, if $i \in \hat{I}_1$), then for every such $\wh{C} \in I(\mc{M})$ 
at least one of the events $\mc{E}_0(\wh{C})$ and $\mc{E}_1(\wh{C})$ holds. Let
\[\mc{E}(\wh{C}) = \mc{E}_0(\wh{C}) \cup \mc{E}_1(\wh{C}).\]
Then, by the definition of $\zeta(C)$ and property (A3) of \cref{assumptions-1}, we have
\begin{align}
\label{eqn:boundf1}
    \mathbf{P}[\mc{E}(\wh{C}) \mid \mf{T}_{\leq i}]
    &\leq 400|\ol{\on{vbl}}(C)|\kappa^{2}\cdot \zeta(C)\cdot \prod_{v \in \ol{\on{vbl}}(C)\setminus \{V^0_i \cup V^+_i\}}(1-3b)^{-\Delta} \p_{\pi}[\on{value}(v)=\boldsymbol{C}_\pi(v)].
\end{align}
Let 
\[\Theta(V^0_i, V^+_i) := \prod_{\mf{C}_i} \left( 400|\ol{\on{vbl}}(C)|\kappa^{2}\cdot \zeta(C)\cdot \prod_{v \in \ol{\on{vbl}}(C)\setminus \{V^0_i \cup V^+_i\}}(1-3b)^{-\Delta} \p_{\pi}[\on{value}(v)=\boldsymbol{C}_\pi(v)] \right),\]
and let
\[\wh{\Theta}_i := \prod_{\mf{C}_i} \left( 400|\ol{\on{vbl}}(C)|\kappa^{2}\cdot \zeta(C)\cdot \prod_{v \in \ol{\on{vbl}}(C)}\left((1-3b)^{-\Delta} \p_{\pi}[\on{value}(v)=\boldsymbol{C}_\pi(v)] + e^{-\kappa/2}\right) \right).\]
Since for any $(C,i), (C',i) \in \mf{C}_i$, we necessarily have $\on{vbl}(C) \cap \on {vbl}(C') = \emptyset$, it follows by expanding the term inside the parentheses in the definition of $\hat{\Theta}_i$ that
\begin{align}
\label{eqn:theta-large}
    \hat{\Theta}_i \geq \sum_{V^0_i, V^+_i \subseteq V, V^0_i \cap V^+_i = \emptyset}\Theta(V_i^0, V_i^+)\cdot e^{-\kappa(|V^0_i|+|V^+_i|)/2}.
\end{align}
Moreover, by property (A2) of \cref{assumptions-1}, we have that
\begin{equation}
\label{eqn:theta-small}
\hat{\Theta}_i \cdot \left(\prod_{\mf{C}_i} |\on{\ol{vbl}}(C)| \right)\leq (3000\Delta)^{-2|\mf{C}_i|}.
\end{equation}
\newline

At this point, we are almost done. Note that
\begin{enumerate}[(N1)]
\item The time intervals $\on{Int}_{j-1}$ for $j \in \hat{J}$ are disjoint from the time intervals $\on{Int}_{i}, \on{Int}_{i+1}, \on{Int}_{i-1}$ for $i \in \hat{I}$. 
\item For each $i \in \hat{I}_0$, we have $f_i = 0$. Recall the definition of the sets $L_1,\dots, L_{\ell} \subseteq \{1,\dots, K-1\}$ associated to $\mc{M}$. Let $j(i)$ be the largest index such that $j(i) \leq i$ and $j(i) \in L_a$ for some $1\leq a \leq \ell$. Let $C^*_{j(i)}$ denote the last constraint in $\mathsf{P}_{j(i)}$. Then, by construction of $I_1(\mc{M})$, we necessarily have $(C^*_{j(i)}, j(i)) \in I_1(\mc{M}) \subseteq I(\mc{M})$. Also, by (M3), there must exist some $v_{j(i)}^* \in \on{\ol{vbl}}(C^*_{j(i)})$ such that $v_{j(i)}^*$ is not updated in    
$(T_{i}, T_{i-1}) = \on{Int}_{i-1}$. 
\item For each $i \in \hat{I}_1$, for each $(C,i) \in I(\mc{M})$, for each $v \in \ol{\on{vbl}}(C)$, at least one of the following holds: (i) $v \in V^+_i$ (ii) $v\in V^0_i$, (iii) there is a term corresponding to $v$ in the expression for $\Theta(\wh{C}, V_i^0, V_i^+)$.
\item By construction, for every $(C,i), (C',j) \in I(\mc{M})$ with $|j-i|\leq 1$, we have $\on{vbl}(C) \cap \on{vbl}(C') = \emptyset$.
\end{enumerate}

\medskip 

\medskip 

Let $\hat{I}_1^{e}$ denote the set of indices $1\leq i \leq K-1$ for which there exists some $j \in \hat{I}_1$ with $|i-j|\leq 1$, and let $\mc{V}$ denote the collection of all subsets $\{V^0_{i-1}, V^+_i\}_{i \in \hat{I}_1^e}$ subject to the restriction that \[|V^+_{i}| + |V^0_{i-1}|  \leq n/2 \quad \quad \quad \forall i \in \hat{I}_1^{e}. \] 
Also, let $\mc{W}$ denote the collection of the variables $v_{j(i)}^*$ (with notation as in (N2)) for each $i \in \hat{I}_0$.
Then, from the above discussion, we have
\begin{align*}
    \mathbf{P}[\mc{M} \cap \{\mf{I} = \hat{J}\}]
    &\leq \sum_{\mc{V}}\sum_{\mc{W}} e^{-8\kappa n|\hat{J}|}\cdot \left(\prod_{i \in \hat{I}_1}\Theta(V^0_i, V^+_i)\cdot e^{-10\kappa\cdot (|V^0_i| + |V^+_i|)}\right)\cdot \prod_{\hat{I}_0}e^{-10\kappa}\\
    &= e^{-8\kappa n |\hat{J}|}\cdot \sum_{\mc{V}}\left(\prod_{i \in \hat{I}_1}\Theta(V^0_i, V^+_i)\cdot e^{-10\kappa\cdot (|V^0_i| + |V^+_i|)}\right)\cdot \sum_{\mc{W}}\prod_{\hat{I}_0}e^{-10\kappa}\\
    &\leq  e^{-8\kappa n |\hat{J}|}\cdot \prod_{i \in \hat{I}_1}\hat{\Theta}_i\cdot \left(\prod_{j\in \{j(i):i\in \hat{I}_0\}} |\ol{\on{vbl}}(C^*_{j})|\right)\cdot 
    \left(\prod_{i \in \hat{I}_0}e^{-10\kappa}\right)\\
    &\leq  e^{-8\kappa n |\hat{J}|}\cdot \prod_{i \in \hat{I}_1}\left(\hat{\Theta}_i \cdot  \left(\prod_{\mf{C}_i} |\on{\ol{vbl}}(C)| \right)\right)\cdot \prod_{i \in \hat{I}_0}e^{-10\kappa} \\
    &\leq e^{-8\kappa n |\hat{J}|}\cdot e^{-10\kappa |\hat{I}_0|}\cdot \prod_{i \in \hat{I}_1}(3000\Delta)^{-2|\mf{C}_i|}\  \\
    &\leq (3000\Delta)^{-2|I(\mc{M})|}\\
    &\leq (3000\Delta)^{-\sum_{i=1}^{\ell}\wh{r}_i}\cdot (3000\Delta)^{-2|I_0(\mc{M})|}.
\end{align*}
Let us explain this chain of inequalities. In the first line, we have used \cref{eqn:exceptional}, \cref{eqn:unexceptional}, \cref{eqn:disjoint-interval}, \cref{eqn:boundf1}, (N1)-(N4), and the law of total probability; the third line follows from \cref{eqn:theta-large} and (N2); the fourth line follows again from (N2); the fifth line follows from \cref{eqn:theta-small}; the sixth line follows upon noting that each of the leg lengths $r_1,\dots, r_{K-1}$ is at most $n$ since $\mathsf{P}_1,\dots, \mathsf{P}_{K-1}$ are induced paths and that $\kappa \geq 4\log(3000\Delta)$ by (A2) on \cref{assumptions-1}; and the last line follows by the construction of $I(\mc{M})$. \qedhere
\end{proof}

\subsection{Rapid mixing of the Glauber dynamics}
We are now in a position to prove \cref{prop:TV-projection}, which will follow as a consequence of the next lemma. 
\begin{lemma} 
\label{lem:disc-check-union}
Let $K \geq 1$ and let $\cal{B}_K$ be the event that there exists a minimal discrepancy check starting from $(v_0,T_0)$ with $T_0$ in time block $K$. Then,
\[
\mathbf{P}(\cal{B}_K) \le 3\Delta K \cdot (4/3)^{-(K-1)}.
\]
\end{lemma}

\begin{proof}
This follows from a union bound. First, note that the graph $G(\mf{C})$ whose vertices are $\mf{C} = \mc{C}\times \{0,1,\dots, K\}$ and which has an edge between $(C,i)$ and $(C',j)$ if and only if $|j-i|\leq 1$ and $\on{vbl}(C) \cap \on{vbl}(C') \neq \emptyset$ has maximum degree at most $3\Delta$. From this, it readily follows that the number of minimal discrepancy checks starting from time block $K$,  with effective parameter $\ell$, and with effective leg lengths $(\wh{r}_1,\dots, \wh{r}_{\ell})$ is at most
\[(3\Delta)\cdot (3\Delta)^{\wh{r}_1+\dots + \wh{r}_{\ell}}\cdot 2^{K-1},\]
where the factor $2^{K-1}$ accounts for the choice of $f_i$ and the first factor of $(3\Delta)$ accounts for the choice of the vertex in $G(\mf{C})$ adjacent to $(C_0, 1)$. By \cref{prop:badprob-check}, we know that for a minimal discrepancy check with effective parameter $\ell$ and effective leg lengths $(\wh{r}_1,\dots, \wh{r}_{\ell})$,
\[\mathbf{P}[\mc{M}] \leq (3000\Delta)^{-\sum_{i=1}^{\ell}\wh{r}_i}\cdot (3000\Delta)^{-2|I_0(\mc{M})|}.\]
Therefore, by the union bound, we have
\begin{align*}
    \mathbf{P}[\mc{B}_K]
    &\leq (3\Delta) \cdot \sum_{\ell = 1}^{K-1}\sum_{\wh{r}_1,\dots, \wh{r}_{\ell}}
    (3000\Delta)^{-\sum_{i=1}^{\ell}\wh{r}_i}\cdot (3\Delta)^{\wh{r}_1+\dots + \wh{r}_{\ell}} \cdot(3000\Delta)^{-2|I_0(\mc{M})|}\cdot 2^{K-1}\\
    &\leq 2^{K-1}\cdot (3\Delta) \cdot \sum_{\ell = 1}^{K-1}\sum_{\wh{r}_1,\dots, \wh{r}_{\ell}}
    2^{-\sum_{i=1}^{\ell}\wh{r}_i}\cdot 500^{-\sum_{i=1}^{\ell}\wh{r}_i} \cdot(3000\Delta)^{-2|I_0(\mc{M})|} \\
    &\leq 2^{K-1}\cdot (3\Delta) \cdot \sum_{\ell = 1}^{K-1}\sum_{\wh{r}_1,\dots, \wh{r}_{\ell}}
    2^{-\sum_{i=1}^{\ell}\wh{r}_i}\cdot 500^{-2|I(\mc{M})|}\\
    &\leq 2^{K-1}\cdot (3\Delta) \cdot \sum_{\ell = 1}^{K-1}\sum_{\wh{r}_1,\dots, \wh{r}_{\ell}}
    2^{-\sum_{i=1}^{\ell}\wh{r}_i}\cdot 500^{-2(K-1)/7}\\
    &\leq (2.95)^{-(K-1)}\cdot (3\Delta) \cdot K\cdot 2^{K-1}\\
    &\leq 3\Delta K \cdot (4/3)^{-(K-1)}.
\end{align*}
Here, the third line follows by the construction of $I(\mc{M})$ and the fourth line follows from (Ind1), (Ind2), and $\sum_{i=1}^{K-1}r_i \geq K-1$. \qedhere

\end{proof}


We are now in a position to prove \cref{prop:TV-projection}.
\begin{proof}[Proof of \cref{prop:TV-projection}]
Consider arbitrary initial states $X_0, Y_0$ and couple the Markov chains $(X_t)_{t\geq 0}, (Y_t)_{t\geq 0}$ as above.
Let 
\[t_* = 10H\log(n\Delta /\delta) + 10n\log(1/\delta),\]
where recall that $H = 100\kappa n$.
If $\tau_{\on{couple}} \geq t_*$
then in particular, there exists some $v \in V$ such that $(v,t_*) \in D$. Let $t_v$ denote the last time before (and including) $t_*$ that $v$ was updated. Note that $(v, t_v) \in U \cap D$. Therefore, by \cref{prop:disc-check-exists}, there exists a minimal discrepancy check $\mc{M}$ starting at $(v,t_v)$.
Note that for any $v \in V$,
\[\mathbf{P}[t_v \leq t_* - 10n\log(1/\delta)] \leq \delta^{10}.\]
On the other hand, if $t_v > t_* - 10n\log(1/\delta)$, then in particular, 
\[t_v \geq 10H\log(n\Delta /\delta)\] so that $t_v$ is in time block $K$ for 
\[K \geq 10\log(n\Delta /\delta).\] 
By \cref{lem:disc-check-union}, the probability of having a minimal discrepancy check starting from some $(v,t_v) \in V\times [t_*]$ with $K-1$ legs is at most
\begin{align*}
6\Delta \cdot n^2\cdot t_*\cdot (4/3)^{-10\log(n\Delta/\delta)} \leq \delta^{5}.
\end{align*}
The desired conclusion now follows from \cref{eqn:mixing}.
\end{proof}

\section{Finishing the proof of \cref{prop:projection}}\label{sec:finish-projection}
We conclude by completing the proof of \cref{prop:projection}.
\begin{proof}[Proof of \cref{prop:projection}]
\noindent {\bf Case 2.} 
Let $\alpha \in [0,1]$ be a parameter to be chosen momentarily via an optimization problem. We will use the marked/unmarked scheme of \cite{moitra2019approximate}. Namely, for each $v\in V$, independently, with probability $\alpha$, we set $Q_v = \{1\}$, and with probability $1-\alpha$, we set $Q_v = [A] = [2]$. Clearly, this satisfies (A3) and (A4). Let $\cal{V}_1$ be the set of $v\in V$ for which $|Q_v|=1$. Let $\cal{V}_f$ be the set of $v\in V$ for which $|Q_v|=A = 2$. 

Let $\theta_1,\theta_f \in (0,1)$ and $\gamma > 0$ be parameters such that 
\begin{enumerate}
    \item $\gamma \leq \theta_1 < \alpha$.
    \item $2\gamma \leq \theta_f < 1-\alpha$.
    \item $D(\theta_1,\alpha) \ge \gamma \log A$.
    \item $D(\theta_f,1-\alpha) \ge \gamma \log A$.
\end{enumerate}
Here, 
\[D(x,y) = x \log(x/y)+(1-x)\log((1-x)/(1-y))\]
is the Kullback-Leibler divergence. Our goal is to maximize $\gamma$. Solving this optimization problem for $A = 2$, we can find parameters $\theta_1, \theta_f, \alpha, \gamma$ such that $\gamma \ge 0.1742$.

Assume that $\Delta \le c A^{\gamma k}/k^2$ for a sufficiently small constant $c$ depending only on $\eta$. We next show that there exists a choice of $\cal{V}_1$ and $\cal{V}_f$ such that for all $C\in \cal{C}$, $|\cal{V}_1 \cap \on{vbl}(C)| \ge \gamma k$ and $|\cal{V}_f \cap \on{vbl}(C)| \ge 2\gamma k$. 

For this, we will use the LLL. Note that if $\mc{V}_1$ is a random set in which each variable is included independently with probability $\alpha$, then by the Chernoff-Hoeffding bound, 
\[\Pr[|\cal{V}_1 \cap C|< \theta_1 |\on{vbl}(C)|] \leq e^{-D(\theta_1,\alpha)|\on{vbl}(C)|}\]
and
\[\Pr[|\cal{V}_f \cap C|< \theta_f |\on{vbl}(C)|] \leq e^{-D(\theta_f,1-\alpha)|\on{vbl}(C)|}.\]
By our assumptions on the parameters (i.e.~they satisfy the conditions of the optimization problem), we have
\[
    \max\left(e^{-D(\theta_1,\alpha)|\on{vbl}(C)|},e^{-D(\theta_f,1-\alpha)|\on{vbl}(C)|}\right) \le e^{-k\log A \cdot \gamma} < (2e\Delta)^{-1}.
\]
Thus, by the LLL, there exists a choice of assignments of $v\in V$ to $\cal{V}_1$ and $\cal{V}_f$ such that for all $C\in \cal{C}$,
\[
|\cal{V}_1 \cap \on{vbl}(C)| \ge \theta_1 k, \quad  |\cal{V}_f \cap \on{vbl}(C)| \ge \theta_f k.
\]
Moreover, by \cref{thm:alg-LLL}, with probability at least $1-\delta$, this assignment can be constructed in time $O(n\Delta k\log(1/\delta))$.

Under a choice of $\cal{V}_1$ and $\cal{V}_f$ with the above properties, for $C\in \cal{C}$, 
\[
b(C) = A^{-|\cal{V}_1 \cap \on{vbl}(C)|} \le A^{-\gamma k} \le \min\left((300\Delta/\eta)^{-1},(600k\Delta)^{-1}\right),
\]
so that (A1) holds. It remains to verify (A2). Note that 
\[
\zeta(C) \le \max(1, A) \le 4.
\]
Let $\kappa = 12\log(3000(k+\Delta))$. Then, we have 
\begin{align*}
    &|\ol{\on{vbl}}(C)|^2 \kappa^2 \cdot \zeta(C) \cdot \prod_{v\in \ol{\on{vbl}}(C)} \left((1-3b)^{-\Delta} \p_\pi[\on{value}(v)=\boldsymbol{C}_\pi(v)] + e^{-\kappa/3}\right)\\
    &\le |\cal{V}_f\cap \on{vbl}(C)|^2 \cdot (12\log(3000(k+\Delta)))^2 \cdot \left(\frac{1+1/(100k)}{A}+\frac{1}{2kA}\right)^{|\cal{V}_f\cap \on{vbl}(C)|} \\
    &\le k^2 \cdot (12\log(3000(k+\Delta)))^2 \cdot \left(\frac{1+1/(2k)}{A}\right)^{|\cal{V}_f \cap \on{vbl}(C)|}\\
    &\le k^2 \cdot (12\log(3000(k+\Delta)))^2 \cdot 2 \cdot A^{-2\gamma k}\\
    &\le (60000\Delta)^{-2},
\end{align*}
by the assumed upper bound on $\Delta$.

\bigskip 

\noindent {\bf Case 3.} 
For each $v\in V$, let $Q_v := \{1,2\}$ and choose a uniformly random projection $\pi_v$ from $[A]$ to $Q_v$ such that $|\pi_v^{-1}(1)|=1$ and $|\pi_v^{-1}(2)|=2$.  Clearly, this satisfies (A3) and (A4). For each $C\in \cal{C}$, let $\cal{V}_1(C)$ be the set of $v\in \on{vbl}(C)$ for which $|\pi_v^{-1}(\boldsymbol{C}_\pi(v))|=1$ and let $\cal{V}_2(C)$ be the set of $v\in \on{vbl}(C)$ for which $|\pi_v^{-1}(\boldsymbol{C}_\pi(v))|=2$. Note that for each $v\in \on{vbl}(C)$, the probability (over the choice of $\pi_v$) that $v\in \cal{V}_1(C)$ is $1/3$ and the probability that $v\in \cal{V}_2(C)$ is $2/3$. 

We have 
\[
b(C) = 2^{-|\cal{V}_2(C)|}. 
\]
Note that 
\[
\zeta(C) \le \max(1, A) = 3.
\]
Let $\kappa = 12\log(3000(\Delta+k))$. Then, we have 
\begin{align*}
    &|\ol{\on{vbl}}(C)|^2 \kappa^2 \cdot \zeta(C) \cdot \prod_{v\in \ol{\on{vbl}}(C)} \left((1-3b)^{-\Delta} \p_\pi[\on{value}(v)=\boldsymbol{C}_\pi(v)] + e^{-\kappa/3}\right)\\
    &\le 3k^2 \cdot (12\log(3000(\Delta+k)))^2 \cdot \left(1+1/(100k)+1/(2k)\right)^{k} \left(\frac{1}{3}\right)^{|\cal{V}_1(C)|}\left(\frac{2}{3}\right)^{|\cal{V}_2(C)|} \\
    &\le 3k^2 \cdot (12\log(3000(\Delta+k)))^2 \cdot 2 \cdot 3^{-k} \cdot 2^{|\cal{V}_2(C)|}.
\end{align*}

Let $\gamma = 0.2$. By Markov's inequality, we have 
\begin{align*}
\Pr[b(C) > 3^{-\gamma k}] &= \Pr[(3/2)^{|\cal{V}_2(C)|} \cdot 3^{|\cal{V}_1(C)|} > 3^{(1-\gamma)k}]\\
&\le \min_{\theta > 0} \left((3^{(1-\gamma) k})^{-\theta} \left(\frac{2}{3} (3/2)^{\theta} + \frac{1}{3} 3^{\theta}\right)^k\right).
\end{align*}
Similarly, we have
\begin{align*}
\Pr[2^{|\cal{V}_2(C)|} > 3^{(1-2\gamma) k}] 
&\le \min_{\theta > 0} \left((3^{(1-2\gamma) k})^{-\theta} \left(\frac{2}{3} 2^{\theta} + \frac{1}{3}\right)^k\right).
\end{align*}
For $\gamma = 0.2$, solving the above optimization problem in $\theta$, one finds that 
\[
\Pr[b(C) > 3^{-\gamma k}] \le 3^{-\gamma k},
\]and 
\[
\Pr[2^{|\cal{V}_2(C)|} > 3^{(1-2\gamma) k}] \le 3^{-\gamma k}.
\]
Using the above bounds, assuming that $\Delta \le c3^{\gamma k}/k^2$, by the LLL, there exists a choice of the projection so that for all $C\in \cal{C}$, 
\[
b(C) \le \min\left((300\Delta/\eta)^{-1}, (600k\Delta)^{-1}\right),
\]and 
\[
|\ol{\on{vbl}}(C)|^2 \kappa^2 \cdot \zeta(C) \cdot \prod_{v\in \ol{\on{vbl}}(C)} \left((1-3b)^{-\Delta} \p_\pi[\on{value}(v)=\boldsymbol{C}_\pi(v)] + e^{-\kappa/3}\right) \le (60000\Delta)^{-2}.
\]
Furthermore, by \cref{thm:alg-LLL}, such projection maps can be constructed in time $O(n\Delta k\log(1/\delta))$ with probability at least $1-\delta$. 

\bigskip 

\noindent {\bf Case 4}. First, we consider the case $A\notin \{5,7\}$. By taking the constant $c > 0$ in the statement of the proposition to be sufficiently small, we may assume that $A^{k}$ is sufficiently large for various inequalities below to go through. 
We will use the inequality 
\[
\inf_{A\ge 4, A\notin \{5,7\}} \left\{ \max_{R\in [2,A]} \left[ \min \left(\frac 1 2 \frac{\log(A/\lceil A/R\rceil)}{\log A}, \frac{\log (\lfloor A/R\rfloor)}{\log A}\right)\right]\right\} := \alpha_* \ge 0.25,
\]
which may be verified numerically. 
For $A \geq 4$, $A\notin \{5,7\}$, let 
\[R := \argmax_{R'\in [2,A]}\left[\min \left(\frac 1 2 \frac{\log(A/\lceil A/R'\rceil)}{\log A}, \frac{\log (\lfloor A/R'\rfloor)}{\log A}\right)\right].\]

As in Case 1, let $Q_{v} = [R]$ for all $v \in V$ and define $\pi_{v}$ arbitrarily so that the preimage of each element in $Q_v$ has size either $\lfloor A/R \rfloor$ or $\lceil A/R \rceil$. As before, this satisfies (A4). 

Recall that $\Delta \le cA^{\alpha_* (k-1)}/(k^2 \log A)$ for some small absolute constant $c$ (depending only on $\eta$). Then, (A1) holds since
\[b  \leq \left(\frac{1}{\lfloor A/R\rfloor}\right)^{k} \leq A^{-\alpha_* k} \le \eta/(300\Delta),\]
and (A3) holds since by the choice of $R$ and since $A$ is sufficiently large,
\[
    \frac{1}{2}A^{-2\alpha_*} \leq \frac{\lfloor A/R\rfloor}{A}  \leq \mb{P}_\pi[\on{value}(v) = \boldsymbol{C}_{\pi}(v)] \leq \frac{\lceil A/R\rceil}{A} \leq A^{-2\alpha_*}.
\]

It remains to verify (A2). By the choice of $R$ and the upper bound on $\Delta$, we have
\[
    (1-3b)^{-\Delta} \le 1+6b\Delta \le 1+1/(100k).
\]
Therefore, as in Case 1, we have that
\[\zeta(C) \leq 3A^{2\alpha_*}.\]
Let $\kappa = 12\log(3000(\Delta+Ak))$. Then, 
\begin{align*}
    &\zeta(C)\cdot \prod_{v\in \ol{\on{vbl}}(C)} \left((1-3b)^{-\Delta} \p_\pi[\on{value}(v)=C_\pi(v)] + e^{-\kappa/3}\right) \\
    &\leq 3A^{2\alpha_*}\cdot (1+1/(2k))^{k}\cdot A^{-2\alpha_*k}\\
    &\leq 6\cdot A^{-2\alpha_*(k-1)}.
\end{align*}
Thus, 
\begin{align*}
    &|\ol{\on{vbl}}(C)|\cdot \kappa^2 \cdot \zeta(C) \cdot \prod_{v\in \ol{\on{vbl}}(C)} \left((1-3b)^{-\Delta} \p_\pi[\on{value}(v)=C_\pi(v)] + e^{-\kappa/3}\right)\\
    &\le 6\cdot k^2 \cdot (12\log(3000(\Delta+Ak)))^{2}A^{-2\alpha_*(k-1)} \\
    &\le (60000\Delta)^{-2},
\end{align*}
where in the last inequality, we used the assumption $\Delta \le cA^{\alpha_*(k-1)}/(k^2 \log A)$. 

Next, we consider the case $A=5$. Here, our construction is similar to Case 3. Let $x,y \geq 0$ be parameters to be chosen later via an optimization problem, which satisfy $x+y = 1$. For each $v\in V$, with probability $x$, we choose the projection $\pi_v$ to be a uniformly random partition of $[A]$ into two parts of sizes $3$ and $2$, and with probability $y$, we choose the projection $\pi_v$ to be a uniformly random partition of $[A]$ into three parts of sizes $2$, $2$, and $1$. For $C\in \cal{C}$, let $\cal{V}_{x,3}(C)$ be the set of variables $v\in \on{vbl}(C)$ for which $\pi_v$ is a partition of $[A]$ into parts of size $3,2$, and $|\pi_v^{-1}(\boldsymbol{C}_\pi(v))|=3$. Similarly, we define $\cal{V}_{x,2}(C)$, $\cal{V}_{y,2}(C)$ and $\cal{V}_{y,1}(C)$. Note that for $v\in \on{vbl}(C)$, $\Pr[v\in \cal{V}_{x,3}(C)] = 3x/5$, $\Pr[v\in \cal{V}_{x,2}(C)] = 2x/5$, $\Pr[v\in \cal{V}_{y,2}(C)] = 4y/5$, and $\Pr[v\in \cal{V}_{y,1}(C)] = y/5$.

We have 
\[
b(C) = 3^{-|\cal{V}_{x,3}(C)|}\cdot 2^{-|\cal{V}_{x,2}(C)|}\cdot 2^{-|\cal{V}_{y,2}(C)|}.
\]
Note that 
\[
\zeta(C) \le \max(1, A) = 5.
\]
Let $\kappa = 12\log(k/c+3000\Delta)$. Then, we have 
\begin{align*}
    &|\ol{\on{vbl}}(C)|^2 \kappa^2 \cdot \zeta(C) \cdot \prod_{v\in \ol{\on{vbl}}(C)} \left((1-3b)^{-\Delta} \p_\pi[\on{value}(v)=\boldsymbol{C}_\pi(v)] + e^{-\kappa/3}\right)\\
    &\le 5k^2 \cdot (12\log(k/c + 3000\Delta))^2 \cdot \left(1+6c/k+1/(2k)\right)^{k} \left(\frac{3}{5}\right)^{|\cal{V}_{x,3}(C)|}\left(\frac{2}{5}\right)^{|\cal{V}_{x,2}(C)|}\left(\frac{2}{5}\right)^{|\cal{V}_{y,2}(C)|}\left(\frac{1}{5}\right)^{|\cal{V}_{y,1}(C)|} \\
    &\le 5k^2 \cdot (12\log(k/c + 3000\Delta))^2 \cdot 2 \cdot 5^{-k} \cdot 3^{|\cal{V}_{x,3}(C)|}\cdot 2^{|\cal{V}_{x,2}(C)|}\cdot 2^{|\cal{V}_{y,2}(C)|}.
\end{align*}

Let $\gamma = 0.221$. By Markov's inequality, we have 
\begin{align*}
\Pr[b(C) > 5^{-\gamma k}] &= \Pr\left[(5/3)^{|\cal{V}_{x,3}(C)|}\cdot (5/2)^{|\cal{V}_{x,2}(C)|}\cdot (5/2)^{|\cal{V}_{y,2}(C)|} \cdot 5^{|\cal{V}_{y,1}(C)|} > 5^{(1-\gamma)k}\right]\\
&\le \min_{\theta > 0} \left((5^{(1-\gamma) k})^{-\theta} \left(\frac{3x}{5} (5/3)^{\theta} + \frac{2x}{5} (5/2)^{\theta} + \frac{4y}{5}(5/2)^{\theta}+\frac{y}{5}5^{\theta}\right)^k\right).
\end{align*}
Similarly, we have
\begin{align*}
\Pr[3^{|\cal{V}_{x,3}(C)|}\cdot 2^{|\cal{V}_{x,2}(C)|}\cdot 2^{|\cal{V}_{y,2}(C)|} > 5^{(1-2\gamma) k}] 
&\le \min_{\theta > 0} \left((5^{(1-2\gamma) k})^{-\theta} \left(\frac{3x}{5} 3^{\theta} + \frac{2x}{5} 2^{\theta} + \frac{4y}{5} 2^{\theta}\right)^k\right).
\end{align*}
For $\gamma = 0.221$, solving the above optimization problem in $\theta$, one finds that 
\[
\Pr[b(C) > 5^{-\gamma k}] \le 5^{-\gamma k},
\]and 
\[
\Pr[3^{|\cal{V}_{x,3}(C)|}\cdot 2^{|\cal{V}_{x,2}(C)|}\cdot 2^{|\cal{V}_{y,2}(C)|} > 5^{(1-2\gamma) k}] \le 5^{-\gamma k}.
\]
Using the above bounds, assuming that $\Delta \le c5^{\gamma k}/k^2$, by the LLL, there exists a choice of the projection so that for all $C\in \cal{C}$, 
\[
b(C) \le \min\left((300\Delta/\eta)^{-1}, (600k\Delta)^{-1}\right),
\]and 
\[
|\ol{\on{vbl}}(C)|^2 \kappa^2 \cdot \zeta(C) \cdot \prod_{v\in \ol{\on{vbl}}(C)} \left((1-3b)^{-\Delta} \p_\pi[\on{value}(v)=\boldsymbol{C}_\pi(v)] + e^{-\kappa/3}\right) \le (60000\Delta)^{-2}.
\]
Furthermore, by \cref{thm:alg-LLL}, such projection maps can be constructed in time $O(n\Delta k\log(1/\delta))$ with probability at least $1-\delta$. 

The case $A=7$ can be done similarly, using partitions of $[A]$ into sets of size $(3,2,2)$ or $(2,2,2,1)$. By using the same analysis as above, one can show that an admissible projection scheme exists when $\Delta \le c7^{\gamma k}/k^2$ for $\gamma = 0.236$, and moreover, that such a projection scheme can be constructed in time $O(n\Delta k\log(1/\delta))$ with probability at least $1-\delta$. 

\bigskip 

\noindent {\bf Case 5.} In the general case, we will combine the constructions used in the previous cases. 

For $v$ with $|\Omega_v|\ge 4$ and $|\Omega_v| \notin \{5,7\}$, let 
\[R_v := \argmax_{R\in [1,|\Omega_v|]} \min \left(\frac 1 2 \frac{\log(|\Omega_v|/\lceil |\Omega_v|/R\rceil)}{\log |\Omega_v|}, \frac{\log (\lfloor |\Omega_v|/R\rfloor)}{\log |\Omega_v|}\right),\]
and recall that 
\[
\min \left(\frac 1 2 \frac{\log(|\Omega_v|/\lceil |\Omega_v|/R_v\rceil)}{\log |\Omega_v|}, \frac{\log (\lfloor |\Omega_v|/R_v\rfloor)}{\log |\Omega_v|}\right) \geq \alpha_* = 0.25.
\] For $v\in V$ with $|\Omega_v| \ge 4$ and $|\Omega_{v}| \notin \{5,7\}$, we let $Q_v = [R_v]$, and define the projection $\pi_v$ arbitrarily so that the preimage of each element in $Q_v$ has size $\lfloor |\Omega_v|/R_v\rfloor$ or $\lceil |\Omega_v|/R_v\rceil$. Let $\cal{V}_L$ be the set of such variables $v\in V$. 

For each $C \in \cal{C}$, define $p_L(C) = \prod_{v\in \cal{V}_L \cap \on{vbl}(C)} \frac{1}{|\Omega_v|}$. Let $\cal{V}_S = V\setminus \cal{V}_L$ and let $p_S(C) = \prod_{v\in \cal{V}_S\cap \on{vbl}(C)} \frac{1}{|\Omega_v|}$. For each $A\in \{2,3,5,7\}$, 
as before, $\on{vbl}_A(C)$ denotes those variables $v \in V$ for which $|\Omega_{v}| = A$. We denote
\[
b_L(C) = \prod_{v\in \cal{V}_L\cap \on{vbl}(C)} \frac{1}{|\pi_v^{-1}(\boldsymbol{C}_\pi(v))|}, \quad \quad b_S(C) = \prod_{v\in \cal{V}_S\cap \on{vbl}(C)} \frac{1}{|\pi_v^{-1}(\boldsymbol{C}_\pi(v))|}.
\]
We also define
\[
t_L(C) = \prod_{v\in \cal{V}_L\cap \on{vbl}(C)} \frac{|\pi_v^{-1}(\boldsymbol{C}_\pi(v))|}{|\Omega_v|}, \quad t_S(C) = \prod_{v\in \cal{V}_S\cap \on{vbl}(C)} \frac{|\pi_v^{-1}(\boldsymbol{C}_\pi(v))|}{|\Omega_v|}, \quad t(C) = t_L(C)t_S(C).
\]
By the definition of $\alpha_*$, we have that 
\[
b_L(C) \le p_L(C)^{\alpha_*}, \quad \text{and} \quad  t_L(C) \le p_L(C)^{2\alpha_*}. 
\]

For $A\in \{2,3,5,7\}$ and for each $v \in V$ with $|\Omega_v| = A$, we define the projection $\pi_v$ randomly using the same partitions of $[A]$ as in previous cases. Note that the choice of the projection depends on certain probability parameters that will be chosen later from an optimization problem. We denote the collection of these parameters by $\mathbf{p}$. Then $b_S(C)^{-1}$ is a product of independent random variables indexed by $v\in \on{vbl}(C) \cap \cal{V}_S$. Similarly, $t_S(C)^{-1}$ is a product of independent random variables indexed by $v\in \on{vbl}(C) \cap \cal{V}_S$. Note that $t_S(C) b_S(C) = p_S(C)$. 

Let $\gamma > 0$ be a parameter to be chosen later from an optimization problem. As in previous cases, we can use Markov's inequality and obtain 
\begin{align*}
\Pr[b(C)^{-1} < p^{-\gamma}] &\le \Pr[b_S(C)^{-1} < p_L(C)^{\alpha_*-\gamma}p_S(C)^{-\gamma}] \\
&= \Pr[b_S(C)/p_S(C) > p_L(C)^{\gamma-\alpha_*}p_S(C)^{\gamma - 1}] \\
&\le \min_{\theta>0}\left(\left(p_L(C)^{\alpha_*-\gamma}p_S(C)^{1-\gamma}\right)^{\theta} \mb{E}[(b_S(C)/p_S(C))^{\theta}]\right).
\end{align*}
Let $\mathfrak{b}_A$ be the random variable $1/|\pi_v^{-1}(\boldsymbol{C}_\pi(v))|$ for (some) $v\in \on{vbl}_A(C)$; note that the distribution of this random variable is the same for all $v \in \on{vbl}_A(C)$ so that $\mathfrak{b}_A$ is indeed well defined. 
Let 
\[
\theta_b = \arg\min_{\theta>0} \max_{A\in \{2,3,5,7\}} \log_A\left(A^{-(1-\gamma)\theta}\mb{E}[(A\mathfrak{b}_A)^{\theta}]\right),
\]and let 
\[
m_b = -\max_{A\in \{2,3,5,7\}} \log_A\left(A^{-(1-\gamma)\theta_b}\mb{E}[(A\mathfrak{b}_A)^{\theta_b}]\right).
\]
Then 
\[
\Pr[b(C)^{-1}<p^{-\gamma}] \le p_L^{(\alpha_*-\gamma)\theta_b} p_S^{m_b} \le p^{\min(m_b,(\alpha_*-\gamma)\theta_b)}.
\]
Similarly, we have 
\begin{align*}
\Pr[t(C)^{-1} < p^{-3\gamma}] &\le \Pr[t_S(C)^{-1} < p_L(C)^{3(\alpha_*-\gamma)}p_S(C)^{-3\gamma}] \\
&= \Pr[b_S(C)^{-1} > p_S(C)^{-1+3\gamma}p_L(C)^{3(\gamma-\alpha_*)}]\\
&\le \min_{\theta>0}\left(\left(p_L(C)^{3(\alpha_*-\gamma)}p_S(C)^{1-3\gamma}\right)^{\theta} \mb{E}[(b_S(C)^{-\theta}]\right).
\end{align*}
Let 
\[
\theta_t = \arg\min_{\theta>0} \max_{A\in \{2,3,5,7\}} \log_A\left(A^{-(1-3\gamma)\theta}\mb{E}[(\mathfrak{b}_A)^{-\theta}]\right),
\]and let 
\[
m_t = -\max_{A\in \{2,3,5,7\}} \log_A\left(A^{-(1-3\gamma)\theta_t}\mb{E}[(\mathfrak{b}_A)^{-\theta_t}]\right).
\]
Then 
\[
\Pr[t(C)^{-1}<p^{-3\gamma}] \le p_L^{3(\alpha_*-\gamma)\theta_t} p_S^{m_t} \le p^{\min(m_t,3(\alpha_*-\gamma)\theta_t)}.
\]

We will maximize $\gamma$ subject to the constraints that 
\begin{align*}
\max_{\mathbf{p}}\left(\min(m_b,(\alpha_*-\gamma)\theta_b)\right) &\ge \gamma,\\
\max_{\mathbf{p}}\left(\min(m_t,3(\alpha_*-\gamma)\theta_t)\right) &\ge \gamma.
\end{align*}
Numerical optimization shows that one can take $\gamma = 0.142$. 

For such $\gamma$ and for sufficiently small $\eta>0$, assuming further that $\Delta \le p^{-\gamma - o_p(1)}$, it follows from the LLL that there exists a choice of projections $\pi_{v}$ for $v \in \mc{V}_S$ so that for all $C\in \cal{C}$, 
\[
b(C) \le (600\Delta/\eta)^{-1},
\]
and
\begin{equation}
t(C) \le p^{3\gamma}. \label{eqn:bound-t}
\end{equation}
Furthermore, by \cref{thm:alg-LLL}, such projection maps can be constructed in time $O(n\Delta k\log(1/\delta))$ with probability at least $1-\delta$. 

For such a projection scheme, the properties (A1), (A3) and (A4) are easily verified. We now show that (A2) is also satisfied. Let $\kappa = 12\log(3000(\Delta+100))$. Since $\zeta(C) \leq 2\Delta$, we have
\begin{align*}
    &|\ol{\on{vbl}}(C)|^2 \kappa^2 \cdot \zeta(C) \cdot \prod_{v\in \ol{\on{vbl}}(C)} \left((1-3b)^{-\Delta} \p_\pi[\on{value}(v)=\boldsymbol{C}_\pi(v)] + e^{-\kappa/3}\right)\\
    &\le (12\log(3000(\Delta+100)))^2 \cdot 2\Delta \cdot |\ol{\on{vbl}}(C)|^2 \prod_{v\in \ol{\on{vbl}}(C)} \left((1-3b)^{-\Delta} \p_\pi[\on{value}(v)=\boldsymbol{C}_\pi(v)] + e^{-\kappa/3}\right).
\end{align*}
By (A1) and our choice of $\kappa$, we have
$(1-3b)^{-\Delta} \p_\pi[\on{value}(v)=\boldsymbol{C}_\pi(v)] + e^{-\kappa/3} \le \frac{3}{4}$ for all $C \in \mc{C}$ and $v\in \on{\ol{vbl}}(C)$. We have two cases.

(1) If there exists some $v \in \on{\ol{vbl}}(C)$ for which $e^{-\kappa/3} > \frac{1}{|\ol{\on{vbl}}(C)|}\p_\pi[\on{value}(v)=\boldsymbol{C}_\pi(v)]$, then 
\begin{align*}
    &(12\log(3000(\Delta+100)))^2 \cdot 2\Delta \cdot |\ol{\on{vbl}}(C)|^2 \prod_{v\in \ol{\on{vbl}}(C)} \left((1-3b)^{-\Delta} \p_\pi[\on{value}(v)=\boldsymbol{C}_\pi(v)] + e^{-\kappa/3}\right)\\
    &\le (12\log(3000(\Delta+100)))^2 \cdot 2\Delta \cdot 2|\ol{\on{vbl}}(C)|^2\cdot (3/4)^{\ol{\on{vbl}}(C)-1}\cdot e^{-\kappa/3}\\
    &\le (60000\Delta)^{-2},
\end{align*}
where the last inequality follows by our choice of $\kappa$.

(2) On the other hand, if for all $v \in \on{\ol{vbl}}(C)$, $e^{-\kappa/3} \le \frac{1}{|\ol{\on{vbl}}(C)|}\p_\pi[\on{value}(v)=\boldsymbol{C}_\pi(v)]$, and further since
\begin{align}
\label{eqn:eta'}
6b\Delta \le \eta/100,
\end{align}
we have
\begin{align}
\label{eqn:wantA2}
    &(12\log(3000(\Delta+100)))^2 \cdot 2\Delta \cdot |\ol{\on{vbl}}(C)|^2 \prod_{v\in \ol{\on{vbl}}(C)} \left((1-3b)^{-\Delta} \p_\pi[\on{value}(v)=\boldsymbol{C}_\pi(v)] + e^{-\kappa/3}\right) \nonumber \\
    &\le (12\log(3000(\Delta+100)))^2 \cdot 6\Delta \cdot |\ol{\on{vbl}}(C)|^2 \prod_{v\in \ol{\on{vbl}}(C)} \left((1-3b)^{-\Delta} \p_\pi[\on{value}(v)=\boldsymbol{C}_\pi(v)]\right) \nonumber \\
    &\le (12\log(3000(\Delta+100)))^2 \cdot \Delta \cdot C_{\eta} \prod_{v\in \ol{\on{vbl}}(C)} \p_\pi[\on{value}(v)=\boldsymbol{C}_\pi(v)]^{1-\eta}.
\end{align}
Here we have used the inequality $x^2 (1-\delta)^x < C_\delta$, which holds for all $\delta \in (0,1)$, $x \geq 1$, and for a sufficiently large $C_\delta$ depending only on $\delta$. 
Now, since $\p_\pi[\on{value}(v) = \boldsymbol{C}_\pi(v)] = \frac{|\pi_v^{-1}(\boldsymbol{C}_\pi(v))|}{|\Omega_v|}$, it follows from \cref{eqn:bound-t} that \cref{eqn:wantA2} is at most $(60000\Delta)^{-2}$, which verifies (A2). \qedhere

\end{proof}

\bibliographystyle{alpha}
\bibliography{main.bib}

\end{document}